\documentclass[11pt]{article}

\usepackage{authblk} % Add this package for handling multiple affiliations

\usepackage[margin=1in]{geometry}
\usepackage{setspace}
\usepackage[makeroom]{cancel}
\usepackage{xcolor}
\usepackage{graphicx}
\usepackage{amsmath}
\usepackage{amssymb}
\usepackage{xspace}
\usepackage[small]{subfigure}
\usepackage[numbers,compress]{natbib}
\usepackage[hyperfootnotes=false]{hyperref}
\usepackage{tcolorbox}

\usepackage{physics}
\usepackage{braket}
\usepackage{tikz}
\usetikzlibrary{calc}
\usepackage[normalem]{ulem}
\usepackage{cleveref}

\newlength{\fighskip} \fighskip=2pt
\newlength{\figvskip} \figvskip=3pt

\newcommand*{\figbox}[2]{{
  \def\figscale{#1}
  \def\arraystretch{0.8}
  \arraycolsep=0pt
  \begin{array}{c}
    \vbox{\vskip\figscale\figvskip
      \hbox{\hskip\figscale\fighskip
        \includegraphics[scale=\figscale]{#2}}}
  \end{array}}}

\usepackage{mciteplus} 
\usepackage{dcolumn}
\usepackage{bm}
\usepackage{verbatim}
\usepackage{amscd}
\usepackage{amsfonts}
\usepackage{setspace}
\usepackage{amsthm}
\usepackage{enumerate}
\usepackage{mathtools}

\theoremstyle{plain}
\newtheorem{theorem}{Theorem}
\theoremstyle{plain}
\newtheorem{lemma}{Lemma}
\newtheorem{fact}{Fact}
\newtheorem{corollary}[lemma]{Corollary}

\DeclareMathOperator{\ED}{ED}

\DeclareMathOperator{\CAP}{Cap}

\DeclareMathOperator{\Area}{Area}

\DeclareMathOperator{\N}{\mathrm{\textbf{N}}}

\DeclareMathOperator{\EPR}{EPR}

\DeclareMathOperator{\dist}{dist}
\newcommand{\prob}[1]{\mathbb{P}\left(#1\right)}
\newcommand{\EX}[1]{\mathbb{E}\left(#1\right)}
\newcommand{\defeq}{\stackrel{\text{def}}{=}}
\newcommand{\calM}{\mathcal{M}}
\newcommand{\st}{\text{s.t.}}

\usepackage[commandnameprefix=ifneeded]{changes}  
\usepackage{xparse}
\NewDocumentCommand{\BY}{o m}{%
  \IfNoValueTF{#1}%
    {\todo[inline, color=purple!40]{\textbf{Beni:} #2}}% set default to inline 
    {\todo[color=purple!40, fancyline]{\textbf{Beni:} #2}}% % default
}
\NewDocumentCommand{\ZL}{o m}{%
  \IfNoValueTF{#1}%
    {\todo[inline, color=blue!40]{\textbf{Zhi:} #2}}% set default to inline 
    {\todo[color=blue!40, fancyline]{\textbf{Zhi:} #2}}% % default
}
\definechangesauthor[name={Zhi}, color=blue]{ZL}

\graphicspath{{Fig/}}

\begin{document}

\title{
%%%%%%report number%%%%%%%%
\vspace{-70pt}
\hfill
{\normalsize YITP-25-15}\\
\hfill
{\normalsize RUP-25-8}\\
\vspace{55pt}
%%%%%%%%%%%%%%%%%%%%%%%%%%%
\bf 
Tripartite Haar random state has no bipartite entanglement}
\author[1,2]{Zhi Li\thanks{Current affiliation: IBM Quantum. zli@ibm.com}}
\author[3,2,4]{Takato Mori\thanks{takato.mori@yukawa.kyoto-u.ac.jp}}
\author[2]{Beni Yoshida\thanks{byoshida@perimeterinstitute.ca}}
\affil[1]{\em \small National Research Council Canada, Waterloo, Ontario N2L 3W8, Canada}
\affil[2]{\em \small Perimeter Institute for Theoretical Physics, Waterloo, Ontario N2L 3W8, Canada}
\affil[3]{\em \small Department of Physics, Rikkyo University, %\protect\\
3-34-1 Nishi-Ikebukuro, Toshima-ku, Tokyo 171-8501, Japan}
\affil[4]{\em \small Yukawa Institute for Theoretical Physics, Kyoto University, \protect\\
Kitashirakawa Oiwakecho, Sakyo-ku, Kyoto 606-8502, Japan} 
\date{}

\maketitle

\begin{abstract}
We show that no EPR-like bipartite entanglement can be distilled from a tripartite Haar random state $|\Psi\rangle_{ABC}$ by local unitaries or local operations when each subsystem $A$, $B$, or $C$ has fewer than half of the total qubits.  
Specifically, we derive an upper bound on the probability of sampling a state with EPR-like entanglement at a given EPR fidelity tolerance, showing a doubly-exponential suppression in the number of qubits.
Our proof relies on a simple volume argument supplemented by an $\epsilon$-net argument and concentration of measure.
Viewing $|\Psi\rangle_{ABC}$ as a bipartite quantum error-correcting code $C\to AB$, this implies that neither output subsystem $A$ nor $B$ supports any non-trivial logical operator.
We also establish general constraints on the structure of tripartite entanglement in Haar random states, showing that W- or GHZ-like entanglement cannot be distilled and that nontrivial global symmetries are absent.
Finally, we discuss a physical interpretation in the AdS/CFT correspondence, indicating that a connected entanglement wedge does not necessarily imply bipartite entanglement, contrary to a previous belief.

\end{abstract}

\tableofcontents

\section{Introduction}

Quantum entanglement lies at the heart of many fundamental questions in quantum physics. 
However, the study of strongly entangled quantum systems poses significant challenges, as analytical and numerical approaches are often limited.
Haar random states provide useful insights into the physical properties of certain many-body quantum systems.  
These states, characterized by their entanglement properties averaged over random ensembles, provide a powerful tool for understanding complex quantum systems with some degree of analytical tractability.
In condensed matter physics, Haar random states (and unitaries) have proven instrumental in understanding dynamical properties, as in eigenstate thermalization hypothesis (ETH)~\cite{Srednicki:1994mfb}, scrambling dynamics~\cite{Hayden:2007cs}, and random matrix theory~\cite{mehta2004random}.
In the AdS/CFT correspondence, Haar random states and unitaries serve as minimal toy models of a quantum black hole~\cite{Page:1993df, Hosur:2015ylk, Cotler:2016fpe}, and Haar random tensor networks serve as toy models obeying the Ryu-Takayanagi formula at the AdS scale at the leading order~\cite{Pastawski:2015qua, Hayden:2016cfa}. 
Also, the properties of Haar random states have played a central role in quantum information theory, as many protocols and fundamental questions rely on these states or on approximations that mimic their behavior~\cite{4262758}.
These examples represent only a fraction of the widespread applications of Haar randomness, which has become an essential concept across diverse areas of modern physics.

Bipartite entanglement of a Haar random state $|\Psi\rangle_{AB}$, when the total system is bipartitioned into two complementary subsystems $A$ and $B$, is well understood~\cite{Page:1993df, Lubkin:1978nch, Lloyd:1988cn}. 
In particular, $|\Psi\rangle_{AB}$ contains $\approx \min(n_A,n_B)$ copies of approximate EPR pairs shared between $A$ and $B$ up to some local unitary (LU) transformations $U_A\otimes U_B$. 
However, the nature of \emph{tripartite entanglement} in a Haar random state $|\Psi\rangle_{ABC}$, where the system is divided into three complementary subsystems $A,B,C$, remains much less understood.
While various entanglement measures and properties of tripartite Haar random states have been studied extensively, the most fundamental question remains open: whether two subsystems $A$ and $B$ in $|\Psi\rangle_{ABC}$ exhibit bipartite EPR-like entanglement or not.

In this paper, we prove that no EPR-like entanglement can be distilled between two subsystems by local unitary (LU) transformations or local operations (LO) when each subsystem $A,B,C$ has fewer than half of the total qubits. 
Specifically, we derive an upper bound with doubly-exponential suppression (in terms of $n$) on the probability of sampling a quantum state with bipartite entanglement. 
Hence, we show that quantum entanglement in a tripartite Haar random state $|\Psi\rangle_{ABC}$ is non-bipartite, despite the fact that the reduced mixed state $\rho_{AB}$ possesses a large amount of quantum (non-classical) correlations (e.g., the mutual information $I(A:B)\sim O(n)$ and the logarithmic negativity $E_{N}(A:B)\sim O(n)$). 
We also discuss an application of our results in the context of quantum error-correcting codes.
Viewing $|\Psi\rangle_{ABC}$ as a random encoding isometry $C \to AB$ with input $C$ and output $AB$, our results imply that each output subsystem $A$ or $B$ supports no logical operator of the code if $n_C<n_A+n_B$ and $|n_A-n_B|<n_C$.

We will also discuss the implications of our results in the context of the AdS/CFT correspondence.
Namely, our results on LU- and LO-distillability suggest that a connected entanglement wedge does not necessarily imply the presence of EPR-like entanglement, contrary to previous beliefs. 
Furthermore, our results on logical operators lead to a surprising prediction concerning entanglement wedge reconstruction: they suggest the possible existence of extensive bulk regions whose degrees of freedom cannot be reconstructed on $A$ or $B=A^c$ when the boundary is bipartitioned into $A$ and $B$. 

\begin{figure}
\centering
a)\includegraphics[width=0.22\textwidth]{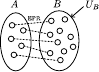}
b)\includegraphics[width=0.28\textwidth]{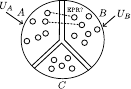}
\caption{a) Bipartite Haar random states with $|A|<|B|$. EPR pairs can be distilled by applying a local unitary $U_B$. 
b) Tripartite Haar random states. Can EPR pairs be distilled by local unitary rotations or local operations?
}
\label{fig_2vs4}
\end{figure}

\subsection{Entanglement in Haar random states}

Consider an $n$-qubit Haar random state $|\Psi_{AB}\rangle$ where qubits are bi-partitioned into two complementary subsystems $A$ and $B$ with $n_A$ and $n_B = n-n_A$ qubits respectively. 
For $n_A < n_B$, we have 
\begin{align}
\mathbb{E}\  \Big\Vert \rho_A - \frac{1}{2^{n_A}}I_A \Big\Vert_1 \lesssim 2^{(n_A-n_B)/2},
\end{align}
where $\mathbb{E}$ represents the Haar average\footnote{
Choosing a Haar random state means that one picks a quantum state uniformly at random from the set of all the $n$-qubit pure states. In particular, a Haar measure can be characterized as a unique probability distribution that is left- and right-invariant under any unitary operators. 
}.
This result is often referred to as Page's theorem~\cite{Page:1993df, Lubkin:1978nch, Lloyd:1988cn}, suggesting that $A$ is nearly maximally entangled with a $2^{n_A}$-dimensional subspace in $B$. 
Namely, there exists a local unitary $I_A \otimes U_B$ acting exclusively on $B$ such that
\begin{align}
(I_A \otimes U_B)|\Psi\rangle_{AB} \approx {|\text{EPR}\rangle^{\otimes n_A}}_{AA'}\otimes |\text{something}\rangle, \qquad |\text{EPR}\rangle = \frac{1}{\sqrt{2}}(|00\rangle + |11\rangle), \label{eq:Page}
\end{align}
where $A'\subseteq B$ and $|A|=|A'|$ with $|R|$ representing the number of qubits in a subsystem $R$. 
Hence, $n_A$ approximate EPR pairs can be distilled between $A$ and $B$ by applying some local unitary transformations without using measurements or classical communications. 
Such unitary transformations can be explicitly constructed by unitarily approximating the Petz recovery map~\cite{Yoshida:2017non}. 
It is worth noting that EPR pairs can be distilled from a single copy of $|\Psi_{AB}\rangle$ without considering the asymptotic (many-copy) scenario.

Next, consider a tripartite $n$-qubit Haar random state $|\Psi_{ABC}\rangle$ on $A$, $B$, and $C$.
We will focus on regimes where each subsystem occupies less than half of the system and thus satisfies
\begin{align}
S_{R} \approx  n_R \qquad \big(0 < n_R < \frac{n}{2}\big) 
\end{align}
for $R=A,B,C$. 
In particular, we will be interested in the asymptotic limit of large $n$. 
Two subsystems, say $A$ and $B$, have a large amount of correlations as seen in the mutual information
\begin{align}
S_{AB} \approx  n_C, \qquad I(A:B) \equiv S_A + S_B - S_{AB} \approx n_{A} + n_{B} - n_{C} \sim O(n).
\end{align}
While the mutual information does not distinguish classical and quantum correlations in general, it can be verified that these are non-classical by computing the logarithmic negativity~\cite{Lu:2020jza, Shapourian:2020mkc}:
\begin{align}
E_{N}(A:B) \approx \frac{1}{2} I(A:B), \qquad E_{N}(A:B) \equiv \log_2 \Big( \sum_j |\lambda_j| \Big),
\end{align}
where $\lambda_j$ are eigenvalues of the partial-transposed density matrix $\rho_{AB}^{T_{A}}$.

A naturally arising question concerns the nature of quantum entanglement in $\rho_{AB}$.
Namely, we will be interested in whether EPR pairs can be distilled in a single copy of $\rho_{AB}$ by applying some local unitary transformation $U_A\otimes U_B$ or some local operation $\Phi_A\otimes \Phi_B$.

When $|\Psi_{ABC}\rangle$ is randomly sampled from (qubit) \emph{stabilizer} states, the mixed state $\rho_{AB}$ contains $\approx \frac{1}{2}I(A:B)$ copies of unitarily (Clifford) rotated EPR pairs. 
This is essentially due to the fact that tripartite qubit stabilizer states have only two types of entanglement: bipartite entanglement (i.e., Clifford rotated EPR pairs) or GHZ-like entanglement~\cite{PhysRevA.81.052302, Nezami:2016zni}, with the latter being rare in random stabilizer states \cite{smith2006typical}.

Do tripartite Haar random states also consist mostly of bipartite entanglement? 
There has been a significant body of previous works addressing this question in a variety of contexts, but no definite/quantitative answer has been provided.
In particular, in the AdS/CFT correspondence, there have been extensive studies based on the bit thread formalism (see~\cite{Freedman:2016zud, Cui:2018dyq, Agon:2018lwq, Harper:2019lff, Bao:2023til} for instance) which, implicitly or explicitly, assume that entanglement in $\rho_{AB}$ are mostly bipartite with $\approx \frac{1}{2}I(A:B)$ EPR pairs.
However, another line of research ~\cite{Akers:2019gcv, Akers:2022max, Akers:2022zxr, Akers:2024pgq} has presented evidence for the presence of genuinely tripartite entanglement in a Haar random state. 

\subsection{Main results: tripartite Haar random state}

Consider a Haar random pure state $|\Psi_{ABC}\rangle$ supported on a $d$-dimensional Hilbert space with $n$ qubits ($d=2^n$). 
In this paper, we prove that no EPR pairs can be distilled from $\rho_{AB}= \Tr_C \big(|\Psi_{ABC}\rangle \langle \Psi_{ABC}| \big)$ via local unitary transformations or local operations when $n_R < \frac{n}{2}$ for $R=A,B,C$ in the large $n$ limit. 

\subsubsection{Local unitaries}
For a non-negative constant $0< h \leq 1$, the (one-shot) LU-distillable entanglement is defined as
\begin{align}\label{def-EDLU}
\ED^{\text{[LU]}}_{h}(A:B) 
\equiv \sup_{m \in \mathbb{N} }\sup_{\Lambda\in \text{LU}} \Big\{ 
 m  \Big| 
\Tr\big( \Lambda(\rho_{AB}) \Pi^{[\text{EPR}]}_{R_{A} R_B} \big) \geq h^2 
  \Big\},
\end{align}
where $\Lambda = U_A \otimes U_B \in \text{LU}$ represents a local unitary acting on $A\otimes B$, and $\Pi^{[\text{EPR}]}_{R_{A} R_B}$ is a projection operator onto $m$ EPR pairs supported on $R_A,R_B$. 
Here, $R_A \subseteq A$ and $R_B\subseteq B$, and $|R_A|=|R_B|=m$ with $|R_A|,|R_B|$ denoting the number of qubits in subsystems $R_A$ and $R_B$ respectively. 
The parameter $h$ controls the fidelity of EPR pairs, where $h =1$ corresponds to perfect EPR pairs while $h \approx 0$ corresponds to low fidelity EPR pairs.

\begin{theorem}\label{theorem:mainLU}
    If $\delta\defeq h^2-2^{-2m}>0$, then for an arbitrary constant $0<c<2$, we have
    \begin{equation}\label{eq:bound}
        \log\prob{\ED^{\text{\emph{[LU]}}}_{h}(A:B) \geq m}
        \leq -c\delta^2d+O\qty((d_A^2+d_B^2)\log \frac{1}{\delta}).
    \end{equation}
\end{theorem}

Note that the assumption $h^2>2^{-2m}$ is necessary. 
In fact, even without applying any unitaries, the region $R_AR_B$ already contains $m$ EPR pairs with fidelity $\sim2^{-2m}$, since $R_AR_B$ is nearly maximally mixed.

Theorem~\ref{theorem:mainLU} provides a meaningful bound when the second term of \cref{eq:bound} is subleading, ensuring that the right-hand side is negative.
The bound then implies that EPR pairs cannot be LU distilled from $\rho_{AB}$ with a fidelity better than that from a maximally mixed state.
In other words, any attempt to enhance the EPR fidelity by applying LU transformations $U_A \otimes U_B$ will be useless! 

Specifically, if we require $\delta$ to be a constant, then it suffices to assume $n_A, n_B=\frac{n}{2}-\omega(1)$ (here $\omega(1)$ means superconstant: a function $f(n)$ is $\omega(1)$ if and only if $\lim_{n\to +\infty} f(n)=+\infty$) which ensures that $d_A^2,d_B^2=o(d)$. 
Namely, 
\begin{itemize}
    \item If $n_A, n_B=\frac{n}{2}-\omega(1)$, then $\prob{\ED^{\text{{[LU]}}}_{h}(A:B) \geq m}\leq \exp(-\Theta(d))$ whenever $h^2-2^{-2m}=\Theta(1)$.
\end{itemize}
Noting that $d=2^n$, we see that the probability of $A$ and $B$ containing EPR-like entanglement is \emph{doubly} exponentially small in the number of qubits.

More generally, an even lower EPR fidelity, allowing $\delta$ to vanish, is also permissible. 
In fact, it suffices to require $\delta>c^\prime \, (d_A+d_B)\sqrt{\frac{\log d}{d}}$ for a sufficiently large constant $c^\prime$, which ensures that the second term remains subleading and the first term diverges to negative infinity.
As long as $n_A, n_B=\frac{n-\log_2 n}{2}-\omega(1)$, we have $(d_A+d_B)\sqrt{\frac{\log d}{d}}=o(1)$, ensuring our choice of $\delta$ can be satisfied.
Therefore,
\begin{itemize}
    \item If $n_A, n_B=\frac{n-\log_2 n}{2}-\omega(1)$, then $\prob{\ED^{\text{{[LU]}}}_{h}(A:B) \geq m}=o(1)$ whenever $h^2-2^{-2m}>c^\prime\, 2^{\max(n_A,n_B)-\frac{n-\log_2 n}{2}}$ for a sufficiently large constant $c^\prime$.
\end{itemize}

\subsubsection{Local operations}

Analogous to \cref{def-EDLU}, the (one-shot) LO-distillable entanglement with fidelity $h^2$ is defined as
\begin{align}
\ED^{\text{[LO]}}_{h}(A:B) 
\equiv \sup_{m \in \mathbb{N} }\sup_{\Lambda\in \text{LO}} \Big\{ 
 m  \Big| 
\Tr\big( \Lambda(\rho_{AB}) \Pi^{[\text{EPR}]}_{R_{A} R_B} \big) \geq h^2 
  \Big\},
\end{align}
where $\Lambda = \Phi_A \otimes \Phi_B$ represents a local operation (channel) acting on $A\otimes B$, and $\Pi^{[\text{EPR}]}_{R_{A} R_B}$ is a projection operator onto $m$ EPR pairs supported on $R_A,R_B$, where $|R_A|=|R_B|=m$.

\begin{theorem}\label{thm-LO}
    If $\delta\defeq h^2-2^{-m}>0$, then for an arbitrary constant $0<c<1$, we have
    \begin{equation}
        \log\prob{\ED^{\text{\emph{[LO]}}}_{h}(A:B) \geq m}
        \leq -c\delta^2d+O(2^{2m}(d_A^2+d_B^2)\log \frac{1}{\delta}).
    \end{equation}
\end{theorem}
The threshold value $2^{-m}$ for $h^2$ is also optimal. In fact, we can consider a simple quantum channel on both sides: attaching $m$ ancillas in $\ket{0}^{\otimes m}$ and doing nothing. 
Making Bell measurement on the ancilla, the probability of getting $m$ EPR pairs is already $2^{-m}$.
Our bound essentially suggests that any attempt to enhance the EPR fidelity cannot outperform this simple quantum channel.  

This result is similar to \cref{theorem:mainLU}, with the threshold value $2^{-2m}$ replaced by $2^{-m}$, and $n_A$ (and $n_B$) replaced by $n_A+m$ (and $n_B+m$) in the second term.
Following the discussions above, we conclude that:
\begin{itemize}
    \item If $n_A, n_B=\frac{n}{2}-m-\omega(1)$, then $\prob{\ED^{\text{{[LU]}}}_{h}(A:B) \geq m}\leq \exp(-\Theta(d))$ whenever $h^2-2^{-m}=\Theta(1)$.
    \item If $n_A, n_B=\frac{n-\log_2 n}{2}-m-\omega(1)$, then $\prob{\ED^{\text{{[LU]}}}_{h}(A:B) \geq m}=o(1)$ whenever $h^2-2^{-m}>c^\prime \, 2^{\max(n_A,n_B)-\frac{n-\log_2 n}{2}+m}$ for a sufficiently large constant $c^\prime$.
\end{itemize}

\subsubsection{Logical operators}

Since the subsystem $C$ is nearly maximally entangled with $AB$ (assuming $n_C<n_A+n_B$), one can view a Haar random state $|\Psi_{ABC}\rangle$ as an approximate isometry $V: C\rightarrow AB$ that encodes $k=n_C$ logical qubits
\begin{align}
V \ : \ \figbox{1.8}{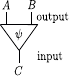}
\end{align}
via the Choi isomorphism, where $C$ ($AB$) corresponds to the input (output) Hilbert space. 

We are interested in whether the encoded quantum information is recoverable from a single subsystem $A$ or not while the complementary subsystem $B$ is traced out (e.g., under erasure errors). 
This question can be addressed by asking whether a \emph{logical operator} can be supported on $A$ or not. 
Loosely speaking, $\overline{U}$ is said to be a logical unitary of $U$ when $\overline{U}$ implements an action of $U$ in the encoded codeword subspace.
If a logical operator $\overline{U}_{A}$ can be supported on $A$, then a piece of information about $U_C$ with respect to the initial state can be deduced from $\rho_A$.

The relation to the LU-distillation problem becomes evident by considering a pair of anti-commuting Pauli logical operators. 
Namely, if Pauli logical operators $\overline{X}, \overline{Z}$ could be supported on $A$, it would imply that an EPR pair could be LU-distilled between $A$ and $C$. 
Hence, one can deduce that logical Pauli operators $\overline{X}, \overline{Z}$ cannot be supported inside $A$. 
This observation enables us to establish the following no-go result on random encoding. 
(A rigorous and quantitative bound will be presented in \cref{thm:LOGICAL}). 

\begin{theorem}[informal]
    Consider a random encoding $C\to AB$. If $n_C<n_A+n_B$ and $n_A<n_B+n_C$, then $A$ contains no quantum information about $C$.
    Namely, $A$ does not support any non-trivial logical unitary operator. \footnote{The inequalities here are informal, serving as the counterpart of the $n_R<n/2$ ($R=A,B,C$) condition in the results for tripartite Haar random states. Precise conditions may be obtained via \cref{thm:LOGICAL} following the discussions below \cref{theorem:mainLU,thm-LO}.}  
\end{theorem}

One interesting corollary of this result is that the so-called cleaning lemma does not necessarily hold for non-stabilizer codes. 
To recap, the cleaning lemma for a stabilizer code asserts that, if a subsystem $A$ supports no non-trivial logical operators, then the complementary subsystem $B=A^c$ supports all the logical operators of the code~\cite{Bravyi:2009zzh}.
This fundamental result is central in establishing the fault-tolerance of topological stabilizer codes (those with geometrically local generators), as it ensures that logical operators can be supported on regions that avoid damaged qubits. 
While the original formulation is restricted to stabilizer codes, analogous properties (e.g. deformability of string logical operators) are known to hold in various models of topological phases beyond the stabilizer formalism.
Despite these examples, our result suggests that the cleaning lemma, in its original formulation given by~\cite{Bravyi:2009zzh}, does not extend to general non-stabilizer quantum error-correcting codes.

\subsubsection{Generalizations}
We also generalize our main theorems in three different directions in section~\ref{sec:generalizations}. 

\begin{enumerate}[i)]
\item \textbf{Tripartite entanglement:} Essentially no nontrivial tripartite pure state of smaller local dimension, such as GHZ, or W-like states, can be LU-distilled from a Haar random tripartite state.
\item \textbf{Chirality:} Tripartite Haar random states are intrinsically chiral. Namely, there is no tripartite LU which converts $|\psi\rangle_{ABC}$ into $|\psi^*\rangle_{ABC}$.
\item \textbf{Global symmetry:} Tripartite Haar random states do not possess any (approximate) global symmetry.
Namely, there is no tripartite LU satisfying $(U_A\otimes U_B \otimes U_C)|\psi\rangle_{ABC}\approx |\psi\rangle_{ABC}$.
\end{enumerate}

\subsection{Miscellaneous comments}

\begin{comment}
Let us present an intuition behind the first proof. 
\ZL[inline]{Now that we added a new subsection for intuition, shall we rearrange this paragraph?}
In a nutshell, we will show that the ``number'' of quantum states with LU-distillable EPR pairs is much smaller than the total ``number'' of quantum states in the Hilbert space. 
However, a set of quantum states in the Hilbert space is not discrete or finite.  
The idea is to consider a discrete set of quantum states, called an $\epsilon$-net, that covers the Hilbert space densely. 
By using states in the $\epsilon$-net as references, we will bound the likelihood of a Haar random state to have LU-distillable EPR pairs. 
A relevant idea was mentioned in our previous work~\cite{Mori:2024gwe}. 
Relying on this, our proof follows from tedious but elementary calculations. 
\end{comment}

\begin{comment}
The downside of this proof, however, is that it does not directly generalize to the LO-distillation. 
This is essentially due to that LOs may be viewed as local unitary operations with ancilla qubits, which spoil the above counting argument.
Our second proof overcomes this problem by observing that a quantum channel with 
$m$ output qubits can always be implemented with at most $2m$ ancilla. Additionally, we leverage the $\epsilon$-net for isometries and the concentration of measure in high-dimensional manifolds.
\end{comment}

Let us mention another relevant question concerning LOCC-distillable entanglement. 
Recall the \emph{hashing lower bound} for LOCC-distillable entanglement in the asymptotic scenario~\cite{Bennett:1996gf,devetak2005distillation}:
\begin{align}
\text{hash}(A:B) \leq  E_{D}(A:B), \qquad \text{hash}(A:B)\equiv \max(S_{A} - S_{AB}, S_{B} - S_{AB}, 0).
\end{align}
In~\cite{Mori:2024gwe}, we considered one-shot 1WAY LOCC-distillable entanglement in $\rho_{AB}$ and showed (see~\cite{Hayden_2006} also)
\begin{align}
E_{D}^{[\text{one-shot 1WAY}]}(A:B) \approx \text{hash}(A:B).
\end{align}
Hence, the hashing lower bound is tight for a tripartite Haar random state under the one-shot 1WAY scenario.

\subsection{Physical and proof intuition}

In this paper, we provide separate proofs for \cref{theorem:mainLU} and \cref{thm-LO}. 
The first proof follows an elementary approach with volume counting;
the second proof requires a few additional prerequisites, but the core idea remains similar. 
In fact, the proof of \cref{thm-LO} can be simplified to establish \cref{theorem:mainLU}.
Since the proofs are somewhat technically involved, it will be helpful to first provide some physical intuition and highlight the key technical ingredients and challenges.

\textbf{Counting argument.}
We begin with the intuition behind the first proof (section~\ref{sec:proof}). 
Consider a tripartite state $|\psi_{ABC}\rangle$ where each subsystem is smaller than half the total size, $n_{A},n_B,n_C < \frac{n}{2}$.
Suppose that, by acting only with local unitaries on $A$ and $B$, we can distill $m$ EPR pairs between them. 
After removing these distilled EPR pairs, the remaining state lives on $n'= n - 2m$ qubits and can be arbitrary:
\begin{align}\label{eq:UaUbpsi}
(U_A \otimes U_B)|\psi\rangle_{ABC} = {\ket{\text{EPR}}^{\otimes m}}_{AB}
\otimes |\text{something}\rangle.
\end{align}
Our goal is to argue that states admitting such a structure form an extremely small subset of the whole Hilbert space.

In an $n$-qubit system, there exist $2^n$ mutually orthogonal pure states.
However, if we relax the requirement of exact orthogonality and only demand that distinct states have small mutual overlaps, one can show that there are in fact doubly-exponentially many such \emph{nearly} orthogonal states.
Namely, letting $\Phi_{\text{state}}(n)$ denote the number of nearly orthogonal $n$-qubit states within some small inner-product cutoff, we have
\begin{equation}
    \Phi_{\text{state}}(n)\sim e^{O(2^n)}.
\end{equation}
Drawing a Haar random state is effectively equivalent to picking a state from this enormous set of quantum states.
Likewise, letting $\Phi_{\text{unitary}}(n)$ denote the number of distinguishable $n$-qubit unitaries, $\Phi_{\text{unitary}}(n)$ is also doubly exponential:
\begin{equation}
    \Phi_{\text{unitary}}(n)\sim e^{O(4^n)}.
\end{equation}
Here the difference in the exponents can be understood via the Choi-Jamiołkowski isomorphism where an $n$-qubit unitary can be viewed as a $2n$-qubit pure state.

Now we can estimate the total number of states with LU-distillable EPR pairs. 
The first step is to pick local unitaries $U_A \otimes U_{B}$ that distill $m$ EPR pairs. 
Since $\Phi_{\text{unitary}}$ grows doubly-exponentially, the number of choices for such unitary operators is upper bounded by 
\begin{align}
\Phi_{\text{unitary}}(n_A) \Phi_{\text{unitary}}(n_B) 
\sim e^{O(4^{n_A})+O(4^{n_B})}
<  \Phi_{\text{unitary}}(n_R + c), \qquad n_{R} = \max\big( n_A, n_B \big)
\end{align}
where $c > 0$ is an $O(1)$ constant. 
Once the EPR pairs are removed, the remaining $n'=n-2m$ qubits can be in an arbitrary state, giving $\Phi_{\text{state}}(n')$ possibilities.
Therefore, the total number of states with LU-distillable EPR pairs is upper bounded by 
\begin{align}
\Phi_{\text{unitary}}(n_R + c) \Phi_{\text{state}}(n') 
\sim e^{O(2^{2n_R + 2c})+O(2^{n-2m})}.
\end{align}
Importantly, if $n_R < \frac{n}{2}$ (more precisely, if $n_R=\frac{n}{2}-\omega(1)$), then:
\begin{equation}
    e^{O(2^{2n_R + 2c})+O(2^{n-2m})}/e^{O(2^{n})}
    < e^{-c'O(2^n)},
\end{equation}
which is doubly-exponentially small. (Here, $c'>0$ is another $O(1)$ constant).
In other words, among the doubly-exponentially large number of nearly orthogonal states in an $n$-qubit Hilbert space, only a doubly-exponentially tiny fraction of states admit LU-distillable entanglement between $A$ and $B$. 
This intuition also clarifies when the counting argument applies: when $2n_R \gtrsim n$ (i.e. one subsystem contains half of the system), the counting argument no longer suppresses such states, and LU-distillation of EPR pairs becomes possible.

\textbf{$\epsilon$-net.}
While the above argument is intuitively appealing, promoting it to a rigorous proof requires overcoming two technical hurdles.
First, the argument's scope is limited to exact EPR pairs (i.e., \cref{eq:UaUbpsi} holds perfectly), but we also hope to consider states that contain approximate EPR pairs, as measured by fidelity.
Second, although Haar measure is defined over a continuous space, the above intuitive argument is based on a discrete set of states.
To address these challenges, we employ the idea of an $\epsilon$-net, which allows us to discretize the continuous Hilbert space. 
Specifically, an $\epsilon$-net is a finite set of reference states, $\mathcal{M}_{\epsilon}$, with the property that any state in the entire space is guaranteed to be within a distance $\epsilon$ of some state in $\mathcal{M}_{\epsilon}$ \cite{Hayden_2004}. 
Elements of an $\epsilon$-net serve as reference points that allow us to characterize the Hilbert space. 

A generic lower bound on the number of states (the cardinality) of an $\epsilon$-net $\mathcal{M}_{\epsilon}$ can be obtained by a simple volume argument. 
However, known constructions contain significantly more states than this lower bound \cite{Hayden_2004}.
In this paper, we will prove the existence of $\mathcal{M}_{\epsilon}$ whose cardinality is essentially tight with respect to this bound (lemma~\ref{lem-net}, Corollary~\ref{corollary:e-net}). 
This construction is crucial and enables us to obtain an essentially tight threshold for LU-distillation results. 
The construction may also be of independent interest.

\textbf{Concentration of measure.}
While the above counting argument provides physically transparent intuition, it unfortunately does not apply to LO-distillation problems.
The main reason is that LO allows ancilla qubits that are initially unentangled, which spoils the volume-counting argument.
Moreover, LO modifies the relevant fidelity threshold significantly, as explained below in \cref{thm-LO}.
As such, we need to use a different strategy.

The key idea is to use concentration of measure.
Many physical quantities computed for Haar random states exhibit concentration of measure, meaning that their probability distributions sharply concentrate around their expectation values.
This phenomenon is most precisely captured by \hyperref[LevyLemma]{Levy’s lemma}, which establishes exponential concentration for observables satisfying a continuity (more precisely, Lipschitz) property.
It turns out that the EPR fidelity satisfies this condition.

We will be interested in the distribution of the EPR fidelity.
For given LOs (characterized by isometries $V_A$ and $V_B$), the expectation value of the EPR fidelity depends on the choice of isometries.
Thus, we cannot directly identify the location of concentration.
The key observation, however, is that we can uniformly upper bound the average fidelity over all possible isometries.
This allows us to conclude that measure concentration occurs at or below the threshold, and therefore the probability of LO distillation is exponentially suppressed.

To extend the argument from fixed isometries to arbitrary local operations, we apply the concentration bound uniformly over an $\epsilon$-net of isometries and then take a union bound.
Although the number of candidate local operations grows exponentially in subsystem dimensions, the concentration estimate is exponentially strong in the total Hilbert space dimension.
Consequently, the union bound does not overwhelm the concentration effect, and the probability that any local operation achieves fidelity above the threshold remains exponentially suppressed.

\section{Local Unitary: Proof of Theorem~\ref{theorem:mainLU}}\label{sec:proof}

\subsection{Spherical cap}

Recall that a pure state can be written as 
\begin{align}\label{eq-defsphere}
|\Psi\rangle = \sum_{j=1}^d  (a_j + ib_j)|j\rangle,\qquad \sum_{j=1}^ da_j^2 + b_j^2 =1
\end{align}
where $a_j, b_j \in \mathbb{R}$. Hence, a pure state corresponds to a point on a unit sphere in $\mathbb{R}^{2d}$. 
Sampling a Haar random state corresponds to choosing a point uniformly at random from a unit sphere\footnote{Here and throughout this paper, we consider pure states with phases, so that we work with the sphere $S^{2d-1}$ rather than the complex projective space $\mathbb{C}P^{d-1}$. Two formulations are equivalent due to the global phase-rotation symmetry.}.
Throughout this paper, a unit sphere in $\mathbb{R}^{p}$ will be called a unit $(p-1)$-sphere and will be denoted by $S^{p-1}$ while its interior is called a unit $p$-ball. Their area and volume are given by
\begin{align}
S_{p-1} = \frac{2 \pi^{p/2}}{\Gamma(\frac{p}{2})},\qquad V_{p} = \frac{\pi^{p/2}}{\Gamma(\frac{p}{2}+1)},
\end{align}
where $\Gamma(z)$ is the Gamma function. 

Consider a unit $(p-1)$-sphere in $\mathbb{R}^p$. 
Letting $(x_1,\cdots, x_p)$ be the coordinates, a spherical cap $\CAP(h)$ is a portion of the sphere cut off by a hyperplane at height $x_1 = h $ where $0 < h \leq 1$, namely:
\begin{align}
\CAP(h) = \Big\{ (x_1,\cdots, x_p) \in \mathbb{R}^p \Big| h\leq x_1 \leq 1, \ \sum_{j=1}^p x_j^2 = 1 \Big\}, 
\end{align}
where the radius of the cap is $r = \sqrt{1-h^2}$. 
See Fig.~\ref{fig_cap}(a) for an illustration.
We will also consider a generalized spherical cap defined as 
\begin{align}
\CAP^{(q)}(h) = \Big\{ (x_1,\cdots, x_p) \in \mathbb{R}^p \Big| h^2 \leq \sum_{j=1}^{q} x_j^2 \leq 1, \ \sum_{j=1}^p x_j^2 = 1 \Big\}
\end{align}
for $1\leq q\leq p-1$.
Note that $\CAP^{(1)}(h)$ consists of two copies of $\CAP(h)$ for $x_1 > 0$ and $x_1 < 0$.

\begin{figure}
\centering
a)\includegraphics[width=0.25\textwidth]{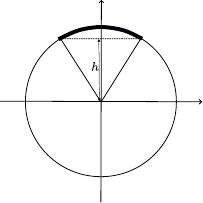}
b)\includegraphics[width=0.25\textwidth]{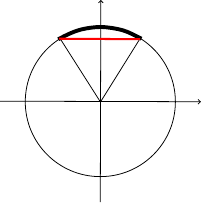}
c)\includegraphics[width=0.25\textwidth]{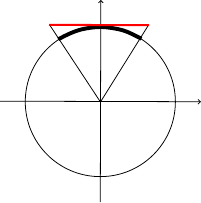}
\caption{A spherical cap and upper/lower bounds on its surface area on a unit sphere $S^{p-1}$. 
Figures are depicted for $p=2$; $x_1$ is the vertical axis. 
a) A spherical cap $\CAP$ in $\mathbb{R}^{p}$. 
b) A lower bound. 
c) An upper bound.
}
\label{fig_cap}
\end{figure}
\subsection{$\epsilon$-net}

For $0 < \epsilon <1$, a set of states $\mathcal{M}_{\epsilon} = \big\{ |\tilde{\Psi}_j\rangle\big \} $ is said to be an $\epsilon$-net in the $2$-norm~\cite{Hayden_2004} when, for every pure state $|\Psi\rangle$ in the Hilbert space, there exists $|\tilde{\Psi}_i\rangle \in \mathcal{M}_{\epsilon}$ such that 
\begin{align}
\big\Vert|\Psi\rangle - |\tilde{\Psi}_i\rangle  \big\Vert_{2} \leq \epsilon.
\end{align} 
See Fig.~\ref{fig_epsilon_net} for an illustration. 
For a $d$-dimensional Hilbert space, the $2$-norm distance between a pair of pure states corresponds to the Euclidean distance in $\mathbb{R}^{2d}$. 
An $\epsilon$-net is often discussed in terms of the $1$-norm in the literature, but an $\epsilon$-net in the $2$-norm suffices to establish the desired result. 
Note that an $\epsilon$-net in the $2$-norm is a $2\epsilon$-net in the $1$-norm~\cite{Hayden_2004}, but the converse is not true.

\begin{figure}
\centering
\includegraphics[width=0.3\textwidth]{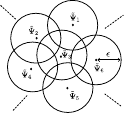}
\caption{A schematic picture of an $\epsilon$-net. Any point in the Hilbert space has some $\epsilon$-neighbor point in the $\epsilon$-net.
}
\label{fig_epsilon_net}
\end{figure}

We begin by proving that there exists an $\epsilon$-net $\mathcal{M}_{\epsilon}$ with $\big|\mathcal{M}_{\epsilon}\big| \lesssim \left(\frac{1}{\epsilon}\right)^{2d} $ up to subleading factors. 
This construction yields an $\epsilon$-net with nearly optimal cardinality, a result that is exponentially tighter (in dimension) than the one in \cite{Hayden_2004}. 
For this, we use the following lemma regarding $\epsilon$-nets on a unit sphere $S^{p-1}$. 

\begin{lemma}\label{lem-net}
Assume $p\geq 3$ and $0<\epsilon<\sqrt{2}$. For any region $R\subseteq S^{p-1}$, there exists an absolute constant $\alpha >0$ and an $\epsilon$-net $\mathcal{M}_\epsilon(R)$ of $R$ such that 
    \begin{equation}
        |\mathcal{M}_\epsilon(R)|\leq \alpha\, p\log p\, \frac{\Area(R^{+2\epsilon})}{\Area(B_\epsilon)}.
    \end{equation}
Here, $B_\epsilon\subseteq S^{p-1}$ is the $\epsilon$ spherical ball (in the Euclidean distance), $R^{+2\epsilon}=\{x\in S^{p-1}|\dist(x,R)\leq 2\epsilon\}$ and $\mathcal{M}_\epsilon(R)\subseteq S^{p-1}$ does not need to be contained in $R$.
\end{lemma}

\begin{proof}
It is proven in \cite{Böröczky2003} that $S^{p-1}$ can be covered by $\epsilon$ spherical balls such that every point on a unit sphere $S^{p-1}$ is covered by less than $400\, p\ln p$ times.
We pick such a covering. 
For any $R\subseteq S^{p-1}$, we collect $\epsilon$ spherical balls that intersect with $R$.
Denote $\mathcal{M}_\epsilon(R)$ as the set of centers of these balls.
Then, $\mathcal{M}_\epsilon(R)$ is an $\epsilon$-net of $R$ since the union of these balls covers $R$.
On the other hand, the union of those balls is contained in $R^{+2\epsilon}$ due to the triangle inequality, and each point of $R^{+2\epsilon}$ is counted by less than $400\, p\ln p$ times, hence
\begin{equation}
        |\mathcal{M}_\epsilon(R)|\Area(B_\epsilon)\leq 400\, p\ln p\Area(R^{+2\epsilon}).
\end{equation}
This completes the proof. 
\end{proof}

Observing that $B_{\epsilon}$ is a spherical cap $\CAP(\sqrt{1-r^2})$ with radius $r=\epsilon\sqrt{1-\frac{\epsilon^2}{4}}$, its area can be lower bounded by
\begin{align}
\Area(B_\epsilon) > \left(\epsilon\sqrt{1-\frac{\epsilon^2}{4}}\right)^{p-1}V_{p-1},
\end{align}
see Fig.~\ref{fig_cap}(b) for an illustration. Hence, by choosing $R= S^{2d-1} = R^{+2\epsilon}$, we obtain the following corollary.

\begin{corollary}\label{corollary:e-net}
For pure states in a $d$-dimensional Hilbert space, there exists an $\epsilon$-net $\mathcal{M}_{\epsilon}$ of $S^{2d-1}$ satisfying 
\begin{align}
\big|\mathcal{M}_{\epsilon}\big| \leq 
2\alpha d\log (2d)\left(\frac{1}{f(\epsilon)}\right)^{2d-1} 
\frac{S_{2d-1}}{V_{2d-1}}
\end{align}
for $0< \epsilon < \sqrt{2}$ where
$f(\epsilon) \equiv \epsilon \sqrt{1- \frac{\epsilon^2}{4}}$.

\end{corollary}

\subsection{State with bare EPR pairs}\label{sec:distil-state}

Consider a set of states where $m$ EPR pairs are already prepared approximately on smaller subsystems $R_A \subseteq A ,R_B \subseteq B$ without the need to apply local unitary transformations. Namely, we define
\begin{align}
\textbf{N}^{(m,h)} \equiv \Big\{  |\Psi\rangle \in \mathcal{H}_{ABC} \Big| 
\langle \Psi | \Pi^{[\text{EPR}]}_{R_{A} R_B}| \Psi\rangle  \geq h^2
 \Big\}
\end{align}
for some fixed $R_A,R_B$ with $|R_A|=|R_B|=m$. 

\begin{lemma}\label{lemma:bare-distillable}
For $0<2\epsilon \leq h \leq 1$, there exists an $\epsilon$-net $\mathcal{N}^{(m,h)}_{\epsilon}$ for \emph{$\textbf{N}^{(m,h)}$} such that
\begin{align}
\big| \mathcal{N}^{(m,h)}_{\epsilon} \big| \leq 
2\alpha d\log (2d)\left(\frac{1}{f(\epsilon)}\right)^{2d-1} 
\frac{
\Area\big(\CAP^{(2\tilde{d})} (h - 2\epsilon) \big)
}{V_{2d-1}},
\end{align}
where $\tilde{d}=2^{n-2m}$.
\end{lemma}

\begin{proof}
Let $R=R_A \cup R_B$ with $|R|=2m$. 
Expand $|\Psi\rangle$ as
\begin{align}
|\Psi\rangle = \sum_{j=1}^{\tilde{d}} \alpha_j |\tilde{1} \rangle_{R} \otimes |j\rangle_{R^c} + 
\sum_{i=2}^{2^{2m}} \sum_{j=1}^{\tilde{d}} \beta_{ij} |\tilde{i} \rangle_{R} \otimes |j\rangle_{R^c}
\end{align}
with orthonormal basis on $R$ and $R^c$ while choosing $|\tilde{1} \rangle_{R} = |\EPR\rangle_{R_AR_B}$.
The state $\ket{\Psi}$ is in $\N^{(m,h)}$ if and only if
\begin{align}
\sum_{j=1}^{\tilde{d}} |\alpha_j|^2 \geq h^2, \qquad \sum_{i=2}^{2^{2m}}\sum_{j=1}^{\tilde{d}} |\beta_{ij}|^2 \leq 1-h^2.
\end{align}
Hence, the total surface region occupied by states in $\N^{(m,h)}$ is $\CAP^{(2\tilde{d})}(h)$ in $\mathbb{R}^{2d}$.
Using lemma~\ref{lem-net} with $R = \CAP^{(2\tilde{d})}(h)$ and observing $R^{+2\epsilon} \subseteq \CAP^{(2\tilde{d})}(h-2\epsilon)$, we obtain the desired result. 
\end{proof}

The area of $\CAP^{(2\tilde{d})}(h)$ has a simple asymptotic expression for large $d$, namely,
\begin{align}
 \frac{\Area\big(\CAP^{(2\tilde{d})}(h)\big)}{S_{2d-1}} \approx u\Big(2^{-2m} - h^2 \Big),
\end{align}
where $u(x)$ is a ``step function'':
for $ h^2 > 2^{-2m}$, the area of $\CAP^{(2\tilde{d})}(h)$ will be negligibly small while for $ h^2 < 2^{-2m}$, the area will be almost as large as that of the unit sphere. 
For our purpose, we only need the estimation when $h^2>2^{-2m}$, which is formalized in the following lemma.
\begin{lemma}\label{lemma:SubGaussian}
Let $\delta \defeq h^2 - 2^{-2m}$. For $\delta >0$,
\begin{align}
\frac{\Area\big(\CAP^{(2\tilde{d})}(h)\big)}{S_{2d-1}} \leq \exp\left(-2(d+1)\delta^2\right). 
\end{align}
\end{lemma}

\begin{proof}
The area of $\CAP^{(2\tilde{d})}(h)$ has an exact expression. Using the polar coordinates on $S^{2d-1}$, we have
\begin{equation}\label{eq-caprela}
    \frac{\Area\big(\CAP^{(2\tilde{d})}(h)\big)}{S_{2d-1}}
    =\frac{1}{B\qty(\tilde{d},d-\tilde{d})}\int_{h^2}^1 x^{\tilde{d}-1}(1-x)^{d-\tilde{d}-1} \,dx,
\end{equation}
where $B(\cdot,\cdot)$ is the Beta function. 
Namely, it equals $\prob{X>h^2}$ where the random variable $X$ follows the Beta distribution Beta$(\tilde{d},d-\tilde{d})$. 
For an $X\sim$ Beta$(\tilde{d},d-\tilde{d})$, we have
\begin{equation}
    \EX{X}=\frac{\tilde{d}}{d}=2^{-2m},~~\text{Var}(X)=\frac{\tilde{d}(d-\tilde{d})}{d^2(d+1)}.
\end{equation}
Therefore, for large $d$, the distribution is narrowly peaked at the mean value $2^{-2m}$.
In particular, Ref.~\cite{marchal2017sub} shows that Beta$(\tilde{d},d-\tilde{d})$ is $\frac{1}{4(d+1)}$-sub-Gaussian, implying that
\begin{equation}\label{eq-subgaussian}
     \frac{\Area\big(\CAP^{(2\tilde{d})}(h)\big)}{S_{2d-1}}
    \leq \exp\left(-2(d+1)\delta^2\right),
\end{equation}
where $\delta = h^2 - 2^{-2m} >0 $.
\end{proof}

\subsection{State with distillable EPR pairs}

Let us define a set of states with LU-distillable EPR pairs
\begin{align}
\hat{\N}^{(m,h)} \equiv \Big\{  |\Psi\rangle \in \mathcal{H}_{ABC} \Big| 
\ED^{\text{[LU]}}_{h}(A:B) \geq m
 \Big\}. 
\end{align}
Note that $\hat{\N}^{(m,h)}$ can be constructed by applying unitaries with the form of $U_A \otimes U_B$ to states in $\N^{(m,h)}$. 

\begin{lemma}\label{lemma:distillable}
Let $\mathcal{N}^{(m,h)}_{\epsilon}$ be an $\epsilon$-net for \emph{$\N^{(m,h)}$}. Let $\hat{\epsilon} \equiv \epsilon + 2\epsilon'$, where $\epsilon'<\sqrt{2}$.
Then, there exists an $\hat{\epsilon}$-net $\hat{\mathcal{N}}^{(m,h)}_{\hat\epsilon}$ for \emph{$\hat{\N}^{(m,h)}$} such that
\begin{align}
\big|  \hat{\mathcal{N}}^{(m,h)}_{\hat{\epsilon}} \big| \leq 
\alpha_1 d_A^2 d_B^2\log (d_A)
\log (d_B)
\left(\frac{1}{f(\epsilon')}\right)^{2d_{A}^2+ 2d_B^2-2}  
\frac{S_{2d_{A}^2-1}S_{2d_{B}^2-1}}{V_{2d_{A}^2-1}V_{2d_{B}^2-1}}
\big|  \mathcal{N}^{(m,h)}_{\epsilon} \big| 
\end{align}
for some absolute constant $\alpha_1>0$.
\end{lemma}

\begin{proof}
    For each $|\Psi\rangle \in \mathcal{N}^{(m,h)}_{\epsilon}$, we can Schmidt decompose it as
    \begin{equation}
        \ket{\Psi}=\sum_{i=1}^{K}\lambda_i\ket{\psi_i}_A\otimes\ket{\phi_i}_{BC},
    \end{equation}
    where the Schmidt rank $K\leq d_A$ since $|A|< |B|+|C|$.
    Therefore, we can represent $U_A\otimes I_{BC}\ket{\Psi}$ as $(\lambda_1 U_A\ket{\psi_1},\lambda_2 U_A\ket{\psi_2},\cdots)$, a normalized complex vector with complex dimension $Kd_A$.  
    Such representation is an isometry from the space $\{U_A\otimes I_{BC}|\Psi\rangle|U_A \in \mathrm{U}(\mathcal{H}_A)\}$ to $S^{2Kd_A-1}$, where $\mathrm{U}(\mathcal{H}_R)$ denotes the group of unitaries acting on a subsystem $R$.
    Conversely, any point on $S^{2Kd_A-1}$ also corresponds to a normalized pure state on $ABC$.
    Therefore, we can use an $\epsilon'$-net on $S^{2Kd_A-1}$ to construct an $\epsilon'$-net of pure states for the space $\{U_A\otimes I_{A^c}\ket{\Psi}\}$.
    The cardinality of this net is upper bounded by $2 \alpha d_A^2 \log (2 d_A^2)
\left(\frac{1}{f(\epsilon')}\right)^{2d_{A}^2-1} \frac{S_{2d_{A}^2-1}}{V_{2d_{A}^2-1}}$ due to corollary~\ref{corollary:e-net} and $K\leq d_A$.

    Repeating the same argument for $U_{B}$ and using the triangle inequality for the $2$-norm, we find a $2\epsilon'$-net for the space $\{U_A\otimes U_B\ket{\Psi}|U_A \in \mathrm{U}(\mathcal{H}_A), U_B \in \mathrm{U}(\mathcal{H}_B)\}$.
    Repeating this procedure for all states $|\Psi\rangle$ in the $\epsilon$-net $\mathcal{N}^{(m,h)}_{\epsilon}$,
    we conclude that there exists an $\hat{\epsilon}$-net $\hat{\mathcal{N}}^{(m,h)}_{\hat\epsilon}$ for $\hat{\N}^{(m,h)}$ whose cardinality is upper bounded as specified in the lemma.
\end{proof}

Combining lemma~\ref{lemma:bare-distillable} and lemma~\ref{lemma:distillable}, we arrive at the following result.

\begin{lemma} \label{lem:netNhat}
For $0<2\epsilon \leq h \leq 1$, $\epsilon'<\sqrt{2}$, there exists an $\hat{\epsilon}$-net with $\hat{\epsilon} = \epsilon + 2\epsilon'$ such that
\begin{equation}
\begin{split}
\big|  \hat{\mathcal{N}}^{(m,h)}_{\hat{\epsilon}} \big| 
\leq & 
\alpha_2 d_A^2 d_B^2 d\log (d_A)
\log (d_B)\log(d)
\left(\frac{1 }{f(\epsilon')}\right)^{2d_{A}^2+2d_B^2-2}  
\left(\frac{1 }{f(\epsilon)}\right)^{2d-1}  \\
&\cdot\frac{
\Area\big(\CAP^{(2\tilde{d})} (h - 2\epsilon) \big)
}{V_{2d-1}} 
\frac{S_{2d_{A}^2-1}S_{2d_{B}^2-1}}{V_{2d_{A}^2-1}V_{2d_{B}^2-1}}
\end{split}
\end{equation}
for some absolute constant $\alpha_2 >0$.
\end{lemma}

\subsection{Putting together}

Recall 
\begin{align}
|\Psi\rangle \in \hat{\N}^{(m,h)} \Rightarrow \exists |\tilde{\Psi}\rangle \in \hat{\mathcal{N}}_{\hat{\epsilon}}^{(m,h)}\ \st \ \big\Vert|\Psi\rangle - |\tilde{\Psi}\rangle  \big\Vert_{2} \leq \hat{\epsilon} .
\end{align}
Hence we have
\begin{align}
\prob { |\Psi\rangle \in \hat{\N}^{(m,h)} }
\leq  \prob { \exists |\tilde{\Psi}\rangle \in \hat{\mathcal{N}}_{\hat{\epsilon}}^{(m,h)} \ \st \ \big\Vert|\Psi\rangle - |\tilde{\Psi}\rangle \big\Vert_{2} \leq \hat{\epsilon}  } .
\end{align}
Let $\{|\tilde{\Psi}_j\rangle\}$ be elements of $\hat{\mathcal{N}}_{\hat{\epsilon}}^{(m,h)}$.
Using the union bound, we have
\begin{align}\label{unionbound1}
\prob { \exists |\tilde{\Psi}\rangle \in \hat{\mathcal{N}}_{\hat{\epsilon}}^{(m,h)} \ \st \ \big\Vert|\Psi\rangle - |\tilde{\Psi}\rangle  \big\Vert_{2} \leq \hat{\epsilon} }
\leq \sum_{j}
\prob {  \big\Vert|\Psi\rangle - |\tilde{\Psi}_{j}\rangle  \big\Vert_{2} \leq \hat{\epsilon} }.
\end{align}

\begin{lemma}\label{lemma:Haar}
Let $|\Psi\rangle$ be chosen according to the Haar measure and $\ket{\Psi_0}$ be an arbitrary, fixed state. 
Let $\epsilon<1$.
Then 
\begin{align}
\prob {  \big\Vert|\Psi\rangle - |\Psi_0\rangle  \big\Vert_{2} \leq \epsilon } <\frac{ g(\epsilon)^{2d-1} V_{2d-1} }{S_{2d-1}},
\end{align}
where $g(\epsilon) \equiv \frac{\epsilon}{\sqrt{1-\epsilon^2}}$.
\end{lemma}

\begin{proof}
The $\epsilon$-ball around $|\Psi_0\rangle$ is given by $\CAP(\sqrt{1-r^2})$ where $r(\epsilon)=\epsilon\sqrt{1 - \frac{\epsilon^2}{4}}$. 
This area can be upper bounded by $(2d-1)$-ball of radius 
\begin{align}
\frac{r(\epsilon)}{\sqrt{1 - r(\epsilon)^2}} < \frac{\epsilon}{\sqrt{1-\epsilon^2}}=g(\epsilon).
\end{align}
See Fig.~\ref{fig_cap}(c) for an illustration. This completes the proof. 
\end{proof}
 
Using \cref{lem:netNhat,lemma:Haar} and \cref{unionbound1}, we have
\begin{equation}
\begin{split}
\prob { |\Psi\rangle \in \hat{\N}^{(m,h)} } <
&\alpha_2 d_A^2 d_B^2 d\log (d_A)
\log (d_B) \log (d) \left(\frac{1 }{f(\epsilon')}\right)^{2d_{A}^2+2d_B^2-2}  
\left(\frac{1 }{f(\epsilon)}\right)^{2d-1}  
g(\hat{\epsilon})^{2d-1}
\\
&
\cdot\frac{S_{2d_{A}^2-1}S_{2d_{B}^2-1}}{V_{2d_{A}^2-1}V_{2d_{B}^2-1}}
\frac{\Area\big(\CAP^{(2\tilde{d})}(h - 2\epsilon) \big) }{S_{2d-1}}.
\end{split}
\end{equation}
Taking the logarithm with base $2$ leads to
\begin{equation}
\begin{split}
\log \Big[  \prob { |\Psi\rangle \in \hat{\N}^{(m,h)} } \Big] < 
\log \qty(\frac{S_{2d_{A}^2-1}}{V_{2d_{A}^2-1}}
\frac{S_{2d_{B}^2-1}}{V_{2d_{B}^2-1}}) 
+ \log \qty( \frac{\Area\big(\CAP^{(2\tilde{d})}(h - 2\epsilon) \big) }{S_{2d-1}}) \\
-(2d_{A}^2+ 2d_{B}^2-2) \log f(\epsilon') 
- (2d-1) \log f(\epsilon) + (2d-1)\log g(\hat{\epsilon}) 
+ O(\log d),
\end{split}
\end{equation}
where the multiplicative factor $\alpha_2 d_A^2 d_B^2 d\log (d_A)
\log (d_B) \log (d)$ leads to a contribution of $O(\log d)$.
Let us split terms into three groups as follows:
\begin{equation}
\begin{split}
A_1 &\equiv \log \qty(\frac{S_{2d_{A}^2-1}S_{2d_{B}^2-1}}{V_{2d_{A}^2-1}V_{2d_{B}^2-1}}), \\
A_2 &\equiv \log \qty( \frac{\Area\big(\CAP^{(2\tilde{d})}(h - 2\epsilon) \big) }{S_{2d-1}}), \\
A_3 &\equiv  - (2d_A^2 + 2d_B^2-2) \log f(\epsilon') - (2d-1) \log f(\epsilon) + (2d-1) \log g(\hat{\epsilon}).
\end{split}
\end{equation}
We then have:
\begin{align}
\log \Big[  \prob { |\Psi\rangle \in \hat{\N}^{(m,h)} } \Big] < A_1 + A_2 + A_3 + O(\log d).  \label{eq:prob_bound}
\end{align}
Note that the bound is valid for arbitrary $\epsilon, \epsilon'$ (as long as $\hat{\epsilon}=\epsilon + 2 \epsilon'<1$). 

Let us evaluate each term. As for $A_1$, we have 
\begin{equation}
\begin{split}
A_1
&=  \log \qty(\frac{2 \pi^{d_A^2}}{\Gamma(d_A^2)}
\frac{2 \pi^{d_B^2}}{\Gamma(d_B^2)} 
\frac{\Gamma(d_A^2  + \frac{1}{2})}{\pi^{d_A^2  - \frac{1}{2}} }
\frac{\Gamma(d_B^2  + \frac{1}{2})}{\pi^{d_B^2  - \frac{1}{2}} })= \log \qty(4\pi \frac{\Gamma(d_A^2  + \frac{1}{2})\Gamma(d_B^2  + \frac{1}{2})}{\Gamma(d_A^2)\Gamma(d_B^2)})\\
&< \log \qty(4\pi d_A d_B ), 
\end{split}
\end{equation}
where we have used $\Gamma(n^2+\frac{1}{2})<n\Gamma(n^2)$.
Hence, we have 
\begin{align}
A_1 < O(\log d) .
\end{align}

As for $A_2$, lemma~\ref{lemma:SubGaussian} implies
\begin{equation}
A_2 < - 2(d+1)\Big((h-2\epsilon)^2 - 2^{-2m}\Big)^2 < -2d(h^2-4\epsilon-2^{-2m})^2 < - 2d \delta^2 + O(d \delta \epsilon),
\end{equation}
where the second inequality follows from $(h-2\epsilon)^2>h^2-4\epsilon$ and assumes $\delta = h^2 - 2^{-2m} > 4\epsilon$. 

As for $A_3$, recall that
\begin{align}
f(\epsilon) \equiv \epsilon \sqrt{1- \frac{\epsilon^2}{16}} = \epsilon(1 + O(\epsilon^2)) ,\qquad g(\epsilon) \equiv \frac{\epsilon}{\sqrt{1-\epsilon^2}}=\epsilon(1 + O(\epsilon^2)).
\end{align}
Hence we have 
\begin{equation}
\begin{split}
A_3 &= 2d \log\frac{g(\epsilon + 2\epsilon')}{f(\epsilon)} + O(d_A^2 \log \frac{1}{\epsilon'}, d_B^2 \log \frac{1}{\epsilon'})\\
&= O\Big(d \frac{\epsilon'}{\epsilon} \Big) + O(d \epsilon^2 ) + O(d \epsilon'^2 ) +  O(d_A^2 \log \frac{1}{\epsilon'}, d_B^2 \log \frac{1}{\epsilon'}).
\end{split}
\end{equation}

Putting these together, for $\delta = h^2 - \frac{\tilde{d}}{d} > 0$, we have 
\begin{equation}
\begin{split}
\log \Big[  \prob { |\Psi\rangle \in \hat{\N}^{(m,h)} }\Big] < \ & - 2d \delta^2  + O\Big(\log d, d\epsilon^2, d\epsilon'^2, d\epsilon \delta, d\frac{\epsilon'}{\epsilon}, d_A^2\log \frac{1}{\epsilon'}, d_B^2\log \frac{1}{\epsilon'} \Big).
\end{split}
\end{equation}
By choosing $\epsilon=\tilde{c}\delta$ and $\epsilon'=\epsilon^2$ where $\tilde{c}$ is a small enough constant, we have\footnote{Here, we omit the $O(\log d)$ term, which is unfortunately unavoidable in this proof due to the $p \log p$ term in \cref{lem-net}. However, it can be avoided if we prove \cref{theorem:mainLU} via the argument used for \cref{thm-LO}. Moreover, in the regimes of interest—when $\delta = \Theta(1)$ and/or $\delta \gg (d_A + d_B)\sqrt{\frac{\log d}{d}}$—the additional $O(\log d)$ term does not affect the asymptotic behavior.}: 
\begin{equation}
\log \Big[  \prob { |\Psi\rangle \in \hat{\N}^{(m,h)} }\Big] < - (2-O(1))d \delta^2  + O\Big( d_A^2 \log \frac{1}{\delta}, d_B^2 \log \frac{1}{\delta} \Big).
\end{equation}
This completes the proof of theorem~\ref{theorem:mainLU}.

\section{Local Operation: Proof of Theorem~\ref{thm-LO}}\label{sec:LO}

In this section, we extend our bound on distillable entanglement to local operations. 
For convenience, we will work within the Stinespring dilation picture, where a local operation is described as first applying an isometry (or attaching ancillas and applying a unitary) and then tracing out a subsystem.
\begin{fact}
    Every quantum channel $\Phi:\mathcal{S}(\mathbb{C}^{d_1})\to\mathcal{S}(\mathbb{C}^{d_2})$ can be expressed as
    \begin{equation}\label{eq:dilation}
        \Phi(\rho)=\Tr_{\mathbb{C}^{d_3}}(V\rho V^\dagger),
    \end{equation}
    where $V: \mathbb{C}^{d_1}\to\mathbb{C}^{d_2}\otimes\mathbb{C}^{d_3}$ is an isometry. 
    Furthermore, since the Kraus rank of $\Phi$ is at most $d_1d_2$ \cite{nielsen2010quantum}, we can always choose $d_3\leq d_1d_2$.
\end{fact}
Applying it to $\Phi_A$ and $\Phi_B$ in our setting, we obtain two isometries $V_A: \mathbb{C}^{d_A}\to\mathbb{C}^{d_A^\prime}$ and $V_B: \mathbb{C}^{d_B}\to\mathbb{C}^{d_B^\prime}$ where\footnote{When constructing $V_A$, we have $d_1=d_A$, $d_2=2^m$ and we denote $d_A^\prime=d_2 d_3$. Similar comments apply to $B$.}
\begin{equation}
    d_A^\prime\leq d_A2^{2m},~~d_B^\prime\leq d_B2^{2m}.
\end{equation}
In other words, crucially, we do not need to attach a large number of ancilla, as the relevant output system on each side contains only $m$ qubits.

Let us first consider the case where the isometries $V_A$ and $V_B$ are fixed.
\begin{lemma}\label{lem-capLO}
    Suppose the quantum channel $\Phi_A$ (and $\Phi_B$) is defined via \cref{eq:dilation} from a given isometry $V_A: \mathbb{C}^{d_A}\to\mathbb{C}^{d_A^\prime}$ (and $V_B: \mathbb{C}^{d_B}\to\mathbb{C}^{d_B^\prime}$, respecitively), $\Lambda = \Phi_A \otimes \Phi_B$, then
       \begin{equation}
            \prob{\Tr\big( \Lambda(\rho_{AB}) \Pi^{[\text{EPR}]} \big) >2^{-m}+\eta}\leq\exp(-d\eta^2).
    \end{equation} 
    Here, the probability is taken over the Haar measure on states $\ket{\psi}$ on $ABC$, $d=2^{n_{ABC}}$, and $\eta>0$ is arbitrary.
\end{lemma}

\newcommand{\img}[1]{\vcenter{\hbox{\includegraphics[height=6em]{#1}}}}   
\newcommand{\tikzimg}[1]{\raisebox{-0.5\height}{\resizebox{!}{6em}{#1}}}
\newcommand{\figa}{
  \begin{tikzpicture}
    % Nodes for tensors
    \node[draw, rectangle, minimum size=0.8cm] (VA) at (0,0) {$V_A$};
    \node[draw, rectangle, minimum size=0.8cm] (VB) at (1.5,0) {$V_B$};
    \node[draw, rectangle, minimum size=0.8cm] (VAp) at (0,2) {$V_A^\dagger$};
    \node[draw, rectangle, minimum size=0.8cm] (VBp) at (1.5,2) {$V_B^\dagger$};
    % Connecting lines
    \draw ($(VAp.south west)+(0.27,0)$) -- ($(VA.north west)+(0.27,0)$);
    \draw ($(VBp.south east)+(-0.27,0)$) -- ($(VB.north east)+(-0.27,0)$);
    \draw ($(VAp.south east)+(-0.27,0)$) to[out=-90,in=-90] ($(VBp.south west)+(0.27,0)$);
    \draw ($(VA.north east)+(-0.27,0)$) to[out=90,in=90] ($(VB.north west)+(0.27,0)$);
    % Trace loop for VA+
    \draw (VAp.north) arc[start angle=0, end angle=180, radius=0.4] 
          -- ++(0,-2.8) 
          arc[start angle=180, end angle=360, radius=0.4] 
          -- (VA.south);
    % Trace loop for VB+ 
    \draw (VBp.north) arc[start angle=180, end angle=0, radius=0.4] 
          -- ++(0,-2.8) 
          arc[start angle=360, end angle=180, radius=0.4] 
          -- (VB.south);
    % Trace loop for C
    \draw (3.4,2.4) arc[start angle=0, end angle=180, radius=0.4] 
          -- ++(0,-2.8) 
          arc[start angle=180, end angle=360, radius=0.4] 
          -- ++(0,2.8);
    % nodes
    \node[below] at (0,-0.7) {$A$};
    \node[below] at (1.5,-0.7) {$B$};
    \node[below] at (3,-0.7) {$C$};
\end{tikzpicture}
}
\newcommand{\figb}{
  \begin{tikzpicture}
    % Nodes for tensors
    \node[draw, rectangle, minimum size=0.8cm] (VA) at (0,0) {$V_A$};
    \node[draw, rectangle, minimum size=0.8cm] (VB) at (1.5,0) {$V_B$};
    \node[draw, rectangle, minimum size=0.8cm] (VAp) at (0,2) {$V_A^\dagger$};
    \node[draw, rectangle, minimum size=0.8cm] (VBp) at (1.5,2) {$V_B^\dagger$};
    % Connecting lines
    \draw ($(VAp.south west)+(0.27,0)$) -- ($(VA.north west)+(0.27,0)$);
    \draw ($(VBp.south east)+(-0.27,0)$) -- ($(VB.north east)+(-0.27,0)$);
    \draw ($(VAp.south east)+(-0.27,0)$) -- ($(VA.north east)+(-0.27,0)$);
    \draw ($(VBp.south west)+(0.27,0)$) -- ($(VB.north west)+(0.27,0)$);
    % Trace loop for VA+
    \draw (VAp.north) arc[start angle=0, end angle=180, radius=0.4] 
          -- ++(0,-2.8) 
          arc[start angle=180, end angle=360, radius=0.4] 
          -- (VA.south);
    % Trace loop for VB+ 
    \draw (VBp.north) arc[start angle=180, end angle=0, radius=0.4] 
          -- ++(0,-2.8) 
          arc[start angle=360, end angle=180, radius=0.4] 
          -- (VB.south);
    % Trace loop for C
    \draw (3.4,2.4) arc[start angle=0, end angle=180, radius=0.4] 
          -- ++(0,-2.8) 
          arc[start angle=180, end angle=360, radius=0.4] 
          -- ++(0,2.8);
    % nodes
    \node[below] at (0,-0.7) {$A$};
    \node[below] at (1.5,-0.7) {$B$};
    \node[below] at (3,-0.7) {$C$};
\end{tikzpicture}
}

\begin{proof}
    We define a function $f_{V_AV_B}(\psi)$ as:
    \begin{equation}
        f_{V_AV_B}(\psi)=\Tr\big( \Lambda(\rho_{AB}) \Pi^{[\text{EPR}]} \big)
        =\bra{\psi}V_A^\dagger V_B^\dagger\Pi^{[\text{EPR}]} V_AV_B\ket{\psi},
    \end{equation}
    where we slightly abuse the notation of $\Pi^{[\text{EPR}]}$: in the third term, it now acts as an operator on $\mathbb{C}^{d_A^\prime}\otimes \mathbb{C}^{d_B^\prime}$.
    We claim $f$ is 1-Lipschitz. In fact, 
    \begin{equation}\label{eq-fLip}
        |f(\psi)-f(\phi)|
        =\left|\Tr(\tilde{\Pi}(\ketbra{\psi}-\ketbra{\phi}))\right|
        % \leq\norm{\tilde{\Pi}}_\infty \norm{\ketbra{\psi}-\ketbra{\phi}}_1
        \leq \norm{\tilde{\Pi}}_\infty \sin\theta(\psi,\phi)
        \leq \norm{\ket{\psi}-\ket{\phi}}_2
        \leq \dist(\psi,\phi).
    \end{equation}
    Here, $\tilde{\Pi}=V_A^\dagger V_B^\dagger\Pi^{[\text{EPR}]} V_AV_B$, $\theta(\psi,\phi)$ is the (Hermitian) angle between $\ket{\psi}$ and $\ket{\phi}$ defined as $\cos\theta(\psi,\phi)=\abs{\braket{\psi|\phi}}$, and $\dist(\psi,\phi)$ is the (geodesic) distance when regarding $\ket{\psi}$ and $\ket{\phi}$ as points on $S^{2d-1}$. 
    The first inequality is due to the fact that the eigenvalues of $\ketbra{\psi}-\ketbra{\phi}$ are $\{\sin\theta,-\sin\theta,0,\cdots,0\}$.

    Let us compute the average of $f_{V_AV_B}(\psi)$ over Haar random $\ket{\psi}$. By standard Haar average computation, we have
    \begin{align}
        \EX{f_{V_AV_B}(\psi)}=\frac{1}{d}\Tr(V_A^\dagger V_B^\dagger\Pi^{[\text{EPR}]} V_AV_B).
    \end{align}
    While it may depend on $V_A$ and $V_B$, the following always holds:
  \begin{align}
        \EX{f_{V_AV_B}(\psi)}=\frac{1}{d}2^{-m}\tikzimg{\figa}\leq\frac{1}{d}2^{-m}\tikzimg{\figb}=2^{-m},
    \end{align}
    where the inequality comes from the fact that $\Tr(XY)\leq\Tr(X)\Tr(Y)$ for two positive semi-definite matrices.
    Therefore, $f_{V_AV_B}(\psi)>2^{-m}+\eta$ implies $f_{V_AV_B}(\psi)>\EX{f_{V_AV_B}(\psi)}+\eta$.

    Now, applying \hyperref[LevyLemma]{Levy’s lemma} (see below) for $f$ on $S^{2d-1}$, we obtain the desired bound:
        \begin{equation}
            \prob{f_{V_AV_B}(\psi)>2^{-m}+\eta}
            \leq \exp(-\frac{2d\eta^2}{2})
            =\exp(-d\eta^2).
    \end{equation}
\end{proof}

In the above proof, we used the following fact about the concentration of measure on high-dimensional spheres \cite{ledoux2001concentration,aubrun2024optimal}.
\begin{fact}[Levy's lemma]\label{LevyLemma}
    If a function $f:S^{p-1}\to\mathbb{R}$ is $K$-Lipschitz in the sense that $|f(x)-f(y)|\leq K\dist(x,y)$ (geodesic distance), then
    \begin{equation}
        \Pr(f>\EX{f}+\eta)\leq \exp(-\frac{p\eta^2}{2K^2}).
    \end{equation}
\end{fact}

Next, to allow arbitrary isometries on $A$ and $B$, we will consider $\epsilon$-nets on the space of isometries.
We denote the space of isometries $\mathbb{C}^{d}\to\mathbb{C}^{d'}$ as $\mathcal{V}_d(\mathbb{C}^{d'})$.
It is a Stiefel manifold $\mathrm{U}(d')/\mathrm{U}(d'-d)$, and $\dim\mathcal{V}_d(\mathbb{C}^{d'})=2dd'-d^2$. 
We quote the following result from Ref.~\cite{szarek1997metric}\footnote{We may simply consider $\epsilon$-nets for $\mathrm{U}(d')$, since an isometry can be (non-uniquely) extended to a unitary. Then \cref{thm-LO} still holds with $2^{2m}$ replaced by $2^{4m}$ in the subleading term.
To obtain a tighter bound on the subleading term, we may also quotient out another $\mathrm{U}(\sqrt{dd'})$, where $dd'$ is the dimension of the subsystem being traced out. This would give us the dimension $dd'-d^2$, which merely improves a constant factor before $2^{2m}$.
}. 
\begin{fact}\label{lem-Vnet}
    There exists an absolute number $c$ such that $\mathcal{V}_d(\mathbb{C}^{d'})$ has an $\epsilon$-net $\mathcal{M}_\epsilon$ in operator norm $\norm{\cdot}_\infty$ such that
    \begin{align}
        |\mathcal{M}_\epsilon|\leq \Big(\frac{c}{\epsilon}\Big)^{2dd'-d^2}.
    \end{align}
\end{fact}

Now we prove \cref{thm-LO}. 
\begin{proof}[Proof of \cref{thm-LO}]

    We pick $\epsilon$-nets for isometries according to fact \ref{lem-Vnet} and denote them as $\mathcal{M}_\epsilon^A$ and $\mathcal{M}_\epsilon^B$ respectively. 
    We will choose $\epsilon$ later.
    For any $\ket{\psi}$, $V_A$ and $V_B$, we can always choose $V_A^\prime\in\calM_\epsilon^A$ and $V_B^\prime\in\calM_\epsilon^B$ such that $\norm{V_A-V_A^\prime}_\infty\leq\epsilon$ and $\norm{V_B-V_B^\prime}_\infty\leq\epsilon$ and hence $\norm{V_AV_B\ket{\psi}-V_A^\prime V_B^\prime\ket{\psi}}_\infty\leq 2\epsilon$.
    It follows from \cref{eq-fLip} that
    \begin{equation}
        f_{V_A^\prime V_B^\prime}(\psi)\geq f_{V_AV_B}(\psi)-2\epsilon.
    \end{equation}
    
    We denote $\delta=h^2-2^{-m}$. 
    By definition, for any $\ket{\psi}$,
    \begin{equation}
        \ED^{\text{{[LO]}}}_{h}(A:B) \geq m 
        \iff
        \exists V_A, V_B \text{ such that } f_{V_AV_B}(\psi)>2^{-m}+\delta.
    \end{equation}
    Therefore, due to the union bound, we have
    \begin{equation}
        \prob{\ED^{\text{[LO]}}_{h}(A:B) \geq m}\leq \sum_{V_A\in \mathcal{M}_\epsilon^A}\sum_{V_B\in \mathcal{M}_\epsilon^B} \prob{f_{V_AV_B}(\psi)>2^{-m}+\delta-2\epsilon}.
    \end{equation}
    Applying fact \ref{lem-Vnet} and \cref{lem-capLO}, we get
    \begin{equation}\label{eq-Vbound}
    \begin{aligned}
        \prob{\ED^{\text{[LO]}}_{h}(A:B) \geq m}
        \leq |\mathcal{M}_\epsilon^A||\mathcal{M}_\epsilon^B|\exp(-d(\delta-2\epsilon)^2)
        \leq\Big(\frac{c}{\epsilon}\Big)^{2d_Ad'_A+2d_Bd'_B}\exp(-d(\delta-2\epsilon)^2),
    \end{aligned}
    \end{equation}
    as long as $\delta>2\epsilon$. 
    
    Now we pick $\epsilon=c^\prime\delta$, where the proportional constant may be taken arbitrarily small. 
    Taking the logarithm of \cref{eq-Vbound}, we obtain:
    \begin{equation}\label{eq:LOfinal}
        \log\prob{\ED^{\text{[LO]}}_{h}(A:B) \geq m}\leq -(1-O(1))\delta^2d+O(2^{2m}(d_A^2+d_B^2)\log \frac{1}{\delta}).
    \end{equation}
\end{proof}

\section{Logical Operators in Random Encoding}\label{sec:LOGICAL}

\subsection{No logical operators in bipartite subsystems}

Given two parameters $0\leq h,w\leq 1$, we say an isometric encoding $V : C \to AB$ admits a logical unitary operator $U_{AB}$ if there exists a unitary $U_C$ such that
\begin{equation}\label{eq:logicalrequire}
    \left| \frac{\Tr(U_C^\dagger V^\dagger U_{AB} V)}{d_C} \right|\geq h^2 \text{~~and~~} \frac{|\Tr(U_C)|}{d_C}\leq w.
\end{equation}

The first inequality imposes a fidelity $h$ on a logical unitary operator $U_{AB}$.
Namely, it compares two isometries $U_{AB} V$ and $VU_C$ via the fidelity between their Choi states:
\begin{equation}
    \frac{\Tr(U_C^\dagger V^\dagger U_{AB} V)}{d_C}
    =\bra{EPR} U_C^\dagger V^\dagger U_{AB} V \ket{EPR}
    =\bra{\Psi}U_{AB}\otimes U_C^*\ket{\Psi},
\end{equation}
where $\ket{\Psi}=(V\otimes I)\ket{EPR}$ is the Choi state for $V$.
Note that
\begin{equation}
    1-|\bra{EPR} U_C^\dagger V^\dagger U_{AB} V \ket{EPR}|\leq \frac{1}{2}\norm{U_{AB} V-VU_{C} }_\infty^2,
\end{equation}
thus, a large $h$ is also a necessary condition for $U_{AB} V$ and $VU_C$ to be close in the operator norm.

The second inequality ensures the non-triviality of the logical operator.  
Namely, $w<1$ is required to ensure that $U_C$ acts non-trivially on the input state.
Otherwise, $w=1$ would imply that $U_C$ is a phase $e^{i\theta}I_C$, and $e^{i\theta}I_{AB}$ is always a trivial logical operator for $e^{i\theta}I_C$ with $h=1$.
On the other hand, if $U_C$ is a unitary conjugation of a Pauli operator, then $w=0$.

\begin{theorem}\label{thm:LOGICAL}
    Assuming $\delta\defeq h^2-w>0$, there exists an absolute constant $c^\prime>0$, such that for a Haar random isometry $V$,
    \begin{equation}
        \log\prob{V\text{~\emph{admits a logical operator on}~} A}
        \leq -c^\prime\delta^2 d+O((d_A^2+d_C^2)\log \frac{1}{\delta}).
    \end{equation}
    Here, $d=2^{n_{ABC}}$.
\end{theorem}
To prove this theorem, we need the following result (\cite{ledoux2001concentration}, sec 2.1) concerning the concentration of measurement on $\mathcal{V}_q(\mathbb{C}^p)$, the space of isometries $\mathbb{C}^q\to \mathbb{C}^p$.
\begin{fact}\label{lem:LevyISO}
    If $f: \mathcal{V}_q(\mathbb{C}^p)\to \mathbb{R}$ is $K$-Lipschitz: $|f(V_1)-f(V_2)|\leq K\norm{V_1-V_2}_F$, where the norm is defined as $\norm{V_1-V_2}_F=\sqrt{\Tr((V_1-V_2)^\dagger (V_1-V_2))}$, then there exist absolute constants $c_1, c_2>0$, independent of $p$ and $q$, such that:
    \begin{equation}
        \Pr(|f-\EX{f}|>\eta)\leq c_1\exp(-\frac{c_2p\eta^2}{K^2}).
    \end{equation}
\end{fact}

\begin{proof}[Proof of \cref{thm:LOGICAL}]
    Fixing unitaries $U_A$ and $U_C$, define a function $f_{U_AU_C}: \mathcal{V}_q(\mathbb{C}^p)\to \mathbb{C}$ as follows:
    \begin{equation}
        (V:\mathcal{H}_C\to \mathcal{H}_{AB}) \mapsto f_{U_AU_C}(V)=\bra{EPR}(V\otimes I)^\dagger U_A\otimes U_C^* (V\otimes I)\ket{EPR}.
    \end{equation}
    Here $q=2^{n_C}$, $p=2^{n_{AB}}$. In the following, we sometimes omit the subscript.
    By standard Haar average computation, we have
    \begin{equation}\label{eq:avefUaUb}
        \EX{f(V)}=\frac{\Tr(U_A)\Tr(U_C^*)}{d_Ad_C}.
    \end{equation}

    We now establish the Lipschitz property of $f$.
    Denote $\ket{\psi_i}=(V_i\otimes I)\ket{EPR}, (i=1,2)$.
    We have
    \begin{equation}
        \braket{\psi_1|\psi_2}=\frac{1}{2^{n_C}}\Tr(V_1^\dagger V_2),
    \end{equation}
    hence
    \begin{equation}
    \begin{aligned}
        \norm{V_1-V_2}_F^2
        &=\Tr(2I_C-V_1^\dagger V_2-V_2^\dagger V_1)=2^{n_C+1}(1-\Re \braket{\psi_1|\psi_2})\\
        &\geq 2^{n_C+1}(1-|\braket{\psi_1|\psi_2}|)\geq 2^{n_C}\sin^2\theta(\psi_1,\psi_2).
    \end{aligned}
    \end{equation}
    In the last step, we used an elementary inequality $1-\cos\theta\geq \tfrac{1}{2}\sin^2\theta$.
    Therefore, similar to \cref{eq-fLip}, we have
    \begin{equation}
        |f(V_1)-f(V_2)|\leq \sin\theta(\psi_1,\psi_2) \leq 2^{-n_C/2}\norm{V_1-V_2}_F.
    \end{equation}
    Therefore, $f$ is $2^{-n_C/2}$-Lipschitz with respect to the Frobenius norm.
    Applying fact \ref{lem:LevyISO} to the real part of $f(V)$,
    we obtain 
    \begin{equation}\label{eq:isoPsmall}
        \prob{\Re f(V)-\Re\EX{f(V)}>\eta}\leq c_1\exp(-\frac{c_2\, 2^{n_{AB}}\, \eta^2}{(2^{-n_C/2})^2})=
        c_1\exp(-c_2d\eta^2),
    \end{equation}
    where $d=2^{n_{ABC}}$.
    
    Now we pick $\epsilon$-nets for $\mathrm{U}(\mathcal{H}_A)$ (and $\mathrm{U}(\mathcal{H}_C)$, respectively) using fact \ref{lem-Vnet}, such that the cardinality is less than $\qty(\frac{c}{\epsilon})^{d_A^2}$ (and $\qty(\frac{c}{\epsilon})^{d_C^2}$, respectively).
    Note that if $\norm{U_A-U_A^\prime}_\infty<\epsilon$ and $\norm{U_C-U_C^\prime}_\infty<\epsilon$, then
    \begin{equation}\label{eq-loc77}
        |f_{U_AU_C}(V)- f_{U_A^\prime U_C^\prime}(V)|<2\epsilon,~~~~\text{for~}\forall V\in \mathcal{V}_q(\mathbb{C}^p).
    \end{equation}
    Taking the expectation over $V$, it also follows that:
    \begin{equation}\label{eq-loc78}
        \left|\mathbb{E}\big(f_{U_AU_C}(V)\big)- \mathbb{E}\big(f_{U_A^\prime U_C^\prime}(V)\big)\right|<2\epsilon.
    \end{equation}
    
    We now assume a given isometry $V$ admits a logical operator that is fully supported on region $A$: there exist $U_A$ and $U_C$ such that $\frac{|\Tr(U_C)|}{d_C}\leq w$ and $|f_{U_AU_C}(V)|\geq h^2$ as in \cref{eq:logicalrequire}.
    We can always assume without loss of generality that $f_{U_AU_C}(V)\geq h^2$ (otherwise, we can multiply $U_A$ or $U_C$ with a phase), which implies that 
    \begin{equation}\label{eq-loc79}
        \Re f_{U_A U_C}(V)-\Re\EX{f_{U_A U_C}(V)}
        \geq h^2-\left|\frac{\Tr(U_A)\Tr(U_C^*)}{d_Ad_C}\right|
        \geq h^2-w.
    \end{equation}
    This inequality, together with \cref{eq-loc77,eq-loc78}, implies that there exist $U_A^\prime$ and $U_C^\prime$ in two $\epsilon$-nets respectively such that 
    \begin{equation}
        \Re f_{U_A^\prime U_C^\prime}(V)-\Re\EX{f_{U_A^\prime U_C^\prime}(V)}>h^2-4\epsilon-w.
    \end{equation}
    Substituting it into \cref{eq:isoPsmall} and applying the union bound over the $\epsilon$-nets, we obtain:
    \begin{equation}
        \prob{V \text{~admits a logical operator on~} A} 
        \leq  2c_1\Big(\frac{c}{\epsilon}\Big)^{d_A^2+d_C^2}\exp(-c_2d(h^2-4\epsilon-w)^2).
    \end{equation}
    Then \cref{thm:LOGICAL} is proved by choosing $\epsilon\propto h^2-w$ with a small enough proportional constant, similar to the proof of \cref{eq:LOfinal}.
\end{proof}

\subsection{Miscellaneous comments}

We have established that, when $|\Psi_{ABC}\rangle$ is a Haar random state (or when $V : C \to AB$ is a random isometry), encoded logical qubits cannot be recovered from subsystems $A$ or $B$ assuming $n_A, n_B, n_C<\frac{1}{2}(n_A+n_B+n_C)$. 
On the contrary, when $|\Psi_{ABC}\rangle$ is a random stabilizer state (or when $V$ is a random Clifford isometry), part of the encoded logical qubits can be recovered from subsystems $A$ and $B$.
Namely, let $g_R$ be the number of independent non-trivial logical operators supported on a subsystem $R$; the following relations are well known~\cite{PhysRevA.81.052302}: 
\begin{align}
g_A = I(A:C)\approx n_A+n_C-n_B, \qquad g_B = I(B:C) \approx n_B+n_C-n_A, \qquad g_A + g_B = 2k
\end{align}
where $k=n_C$ in our setting and the approximations hold when $n_R<\frac{n}{2}$ for $R=A,B,C$. 
This suggests that $I(A:C)\sim O(n)$ logical operators can be supported on $A$.
In fact, when $V$ is sampled randomly, it is likely that a given logical operator $\ell_A$ on $A$ can find some other logical operator $r_A$ on $A$ that anti-commutes with $\ell_A$. 
As such, one can choose $\frac{g_A}{2} - o(1)$ pairs of mutually anti-commuting basis logical operators on $A$, suggesting that $\frac{g_A}{2} - o(1)$ logical qubits can be recovered from $A$. 

In the above discussion, it is crucial to consider unitary logical operators $\overline{U_C}$. 
In fact, some \emph{non-unitary} logical operators can be constructed on $A$ or $B$. 
For instance, let us split $C$ further into two subsystems $C=C_{0}C_1$ where $C_0$ consists only of one qubit while $C_1$ consists of $n_C-1$ qubits. 
Consider the following operator
\begin{align}\label{eq:nonUlogical}
V_{C} = X_{C_0} \otimes |0\rangle\langle 0 |_{C_{1}}.
\end{align}
where $|0\rangle\langle 0 |_{C_{1}}$ is a projection operator acting on $C_1$. 
Observe that $|0\rangle\langle 0 |_{C_{1}}$ acting on a Haar random state $|\Psi_{ABC}\rangle$ effectively creates another Haar random state $|\Phi_{ABC_0}\rangle \propto  |0\rangle\langle 0 |_{C_{1}} |\Psi_{ABC}\rangle$ after an appropriate normalization. 
Then, finding a logical operator $\overline{V_{C}}$ for  $|\Psi_{ABC}\rangle$ is equivalent to finding a logical operator $\overline{X_{C_0}}$ for  $|\Phi_{ABC_0}\rangle$, reducing the problem to an EPR distillation in the projected state $|\Phi_{ABC_0}\rangle$.  
Then, if $A$ contains more than half of $ABC_{0}$, namely $n_{A} > n_{B} +1$, $\overline{X_{C_0}}$ can be supported on $A$ even when no logical \emph{unitary} operator can be supported on $A$.

We note that, in fact, \cref{thm:LOGICAL} also holds for non-unitary operators with bounded operator norm. 
The inability to exclude non-unitary logical operators stems from the fidelity threshold $h$: 
for $V_C$ in \cref{eq:nonUlogical}, we have $h^2 =O(1/d_C)$, rendering the bound in \cref{thm:LOGICAL} vacuous when $n_A > n_B + 1$. 
Referring to \cref{eq:nonUlogical} as a non-unitary logical operator implicitly assumes post-selection on the $|0\rangle\langle 0|$ outcome, which effectively rescales the fidelity threshold by the success probability of the measurement.

It is worth noting that the reconstruction of a non-unitary logical operator $\overline{V_C}$ on $A$ can be interpreted as (one-shot) LOCC entanglement distillation where i) one performs projective measurements $\{|i\rangle\langle i| \}_{C_1}$ on $C_1$, ii) sends the measurement outcome $i$ to $A$, and iii) applies an appropriate LU on $A$ to prepare an EPR pair between $A$ and $C_{0}$. See~\cite{Mori:2024gwe} for details.

\section{Generalizations}\label{sec:generalizations}

The method developed in this paper is general and admits broader applicability. 
In this section, we establish three theorems highlighting the tripartite entanglement structure of Haar random tripartite states.
We omit some technical details in this section, as the proofs are similar to those in \cref{sec:LO,sec:LOGICAL}.

\subsection{Nothing is LU-distillable}

We have shown that bipartite EPR-like entanglement is essentially absent in tripartite Haar random states.
This naturally raises a broader question: can any other tripartite pure state be LU-distilled?
For instance, could one distill GHZ-like or W-like entanglement?

We now show a much stronger no-go statement: essentially no nontrivial tripartite pure state of smaller local dimension can be LU-distilled from a tripartite Haar random state.
Let the tripartition be $A,B,C$, and decompose each subsystem as
\begin{align}
A = A_{1} \otimes A_2, \quad 
B = B_{1} \otimes B_2, \quad 
C = C_{1} \otimes C_2.
\end{align}
% Let $|\phi\rangle_{A_1 B_1 C_1}$ be an arbitrary tripartite pure state supported on $A_1B_1C_1$. 
We ask whether there exists a local unitary $U = U_A\otimes U_B\otimes U_C$ and a pure state $|\phi\rangle_{A_1 B_1 C_1}$ such that
\begin{align}\label{eq:LU_distill_phi}
 % \qquad \text{such that} \qquad 
U |\psi\rangle_{ABC} = |\phi\rangle_{A_1 B_1 C_1} \otimes |\text{something}\rangle_{A_2B_2 C_2}.
\end{align}
Without loss of generality, we may assume $1\leq n_1\leq n_2$ (here $n_i=n_{A_i}+n_{B_i}+n_{C_i}$). 
If $n_1\geq n_2$, we can consider an equivalent problem of outputting a pure state $|\phi'\rangle$ on $A_2B_2C_2$.
% if $n_1\leq 2$, then the problem is at most bipartite.

\begin{theorem}
[Nothing is LU-distillable]\label{thm:nothing}
Let $|\psi\rangle_{ABC}$ be a Haar random state with $n_A,n_B, n_C < n/2$. 
% Let $|\phi\rangle_{A_1 B_1 C_1}$ be an arbitrary tripartite pure state with 
Then, the probability that there exists a tripartite pure state $|\phi\rangle_{A_1 B_1 C_1}$ which can be LU-distilled from $|\psi\rangle_{ABC}$ is $o(1)$ in the large-$n$ limit.
\end{theorem}

\begin{proof}[Proof Sketch]
    %Decompose the Hilbert spaces as $H_A=H_{A_1}\otimes H_{A_2}$ (and same for $B,C$).
    Consider the entanglement entropy of $\ket{\psi}$ with respect to the bipartition into $A_1B_1C_1$ and $A_2B_2C_2$, denoted as $S(\psi)$.
    If there exists a LU-distillable tripartite pure state $|\phi\rangle$ on $A_1B_1C_1$ (namely, \cref{eq:LU_distill_phi} holds), then $S(U\psi)=0$.
    Here, this entropy is across the $(A_1B_1C_1):(A_2B_2C_2)$ cut.
    In the following, we relax this condition to $S(U\psi)<\delta$, and upper bound its probability.
    %  (with $\delta<1$)
    % Without loss of generality, we may assume $n_1\leq n_2$. 
    % (Otherwise, we can consider an equivalent problem of outputting a pure state $|\phi'\rangle$ on $A_2B_2C_2$.)
    
    Page's theorem shows that $\mathbb{E}_{\psi}(S(U\psi))=\mathbb{E}_{\psi}(S(\psi))\approx n_1$ for any $U$.
    Moreover, we have the following lemma regarding the Lipschitz continuity of the entanglement entropy:
    \begin{lemma}
        $S(\psi)$ is $O(n_1)$-Lipschitz as a function of $\psi$. 
    \end{lemma}
    \begin{proof}
        Given $\ket{\psi}$ and $\ket{\phi}$, we consider the two-dimensional subspace spanned by them and define $H$ as a rotation that rotates $\ket{\psi}$ to $\ket{\phi}$ in unit time.
        We have $\norm{H}_\infty=\theta(\psi,\phi)$, where $\cos\theta(\psi,\phi)=\abs{\braket{\psi|\phi}}$.
        Ref.~\cite{bravyi2007upper} shows that
        \begin{equation}
            \frac{dS(e^{-iHt}\psi)}{dt}\leq c\norm{H}_\infty\log_2 \min\{d_1,d_2\},
        \end{equation}
        where $c<1.9123$.
        Hence $|S(\psi)-S(\phi)|\leq cn_1\theta(\psi,\phi)$.
    \end{proof}
    \hyperref[LevyLemma]{Levy’s lemma} then implies that $\prob{S(\psi)<\delta}$ is doubly-exponentially small when $\delta$ is bounded well below $n_1$ (say, $\delta \leq n_1 - c\sqrt{n_1} $ with some $c>0$):
    \begin{equation}
        \prob{S(\psi)<\delta}\lesssim \exp(-\frac{d}{n_1}) \lesssim \exp(-\frac{d}{n}),~~~d=2^n.
    \end{equation}
    Covering the space of LUs by an $\epsilon$-net with size $\sim (\frac{1}{\epsilon})^{d_A^2+d_B^2+d_C^2}$ and applying the union bound, we can establish that, with overwhelming probability, no smaller tripartite pure state $|\phi\rangle_{A_1 B_1 C_1}$ can be LU-distilled.    
\end{proof}

This result shows that as $n_A,n_B,n_C$ increase (while remaining below $n/2$), a Haar random state represents a genuinely new form of tripartite entanglement that cannot be reduced, by local unitaries, to any smaller tripartite pure state. In this sense, tripartite Haar random states are maximally irreducible under LU-distillation.\footnote{
While this result concerns LU-distillation of pure states supported on smaller tripartite subsystems,
it also has implications for the LU-distillation of certain target states supported on the original tripartite systems.
For example, it implies that the probability for a tripartite state to be LU-equivalent to a GHZ state vanishes,
since a GHZ state over $ABC$ is the tensor product of two GHZ states over $A_1 B_1 C_1$ and $A_2 B_2 C_2$
(the same conclusion also follows from a direct counting argument).
}

\subsection{Chirality of wavefunction}

Another interesting application of our method concerns the intrinsic chirality of wavefunctions. 
Conventionally, chirality has been characterized dynamically through responses to external perturbations, such as robust edge currents or quantized transport coefficients. 
Recently, there have been significant developments toward characterizing chirality directly from a single many-body wavefunction~\cite{kim2022chiral}.
In this subsection, we will focus on the characterization based on multipartite local unitaries \cite{vardhan2026chirality}. 
More precisely, an $m$-partite pure state $|\psi\rangle_{R_1\cdots R_m}$ is said to be \emph{$m$-partite non-chiral} if there exists a local unitary of the form 
\begin{align}
U = U_{R_1}\otimes \cdots \otimes U_{R_n} \qquad \text{such that} \qquad U |\psi\rangle = |\psi^* \rangle,
\end{align}
where $|\psi^* \rangle$ denotes the complex conjugate of $|\psi\rangle$ in the computational basis. Each component $U_{R_j}$ may be arbitrary. If no such local unitary exists, the state is called \emph{$m$-partite chiral}.

We now ask whether Haar random states are chiral in this sense.
It is instructive to first consider the bipartite case. A Haar random state $|\psi\rangle_{AB}$ is not bipartite-chiral. Indeed, by Page’s theorem (\cref{eq:Page}), it can be converted by bipartite local unitaries into a tensor product of approximate EPR pairs and a product state. Since each EPR pair may be chosen to be real, the resulting state can be mapped to its complex conjugate by further local unitaries. Thus, bipartite Haar random states are generically non-chiral.

The situation changes dramatically in the tripartite setting.

\begin{theorem}[Tripartite chirality]\label{thm:chiral}
    Let $|\psi\rangle_{ABC}$ be a Haar random state with $n_A,n_B,n_C < n/2$. Then $|\psi\rangle_{ABC}$ is tripartite chiral with probability $1-o(1)$ in the large-$n$ limit.
\end{theorem}

\begin{proof}[Proof Sketch]
Suppose there exists a tripartite local unitary $U=U_A\otimes U_B\otimes U_C$ such that $\ket{\psi^*}=U\ket{\psi}$,
or equivalently 
\begin{align}
\mel{\psi^*}{U}{\psi}=1.
\end{align}
Define
\begin{align}
f_U(\psi):=\mel{\psi^*}{U}{\psi}.
\end{align}
One verifies that $f_U$ is $O(1)$-Lipschitz on the unit sphere $S^{2d-1}$.

We next evaluate the Haar average of $f_U(\psi)$. The Haar measure is invariant under global phase rotations $|\psi\rangle \mapsto e^{i\theta} |\psi\rangle$. Under this transformation, we have
\begin{align}
f_U(\psi) \mapsto e^{i2\theta} f_U(\psi).
\end{align}
Hence,
\begin{align}
\mathbb{E}_\psi \big[ f_U(\psi) \big] =
e^{i2\theta}
\mathbb{E}_\psi \big[ f_U(\psi) \big]
\end{align}
for all $\theta$, which implies
\begin{align}
\mathbb{E}_{\psi} \bigl[ f_U(\psi) \bigr] = 0.
\end{align}

Applying \hyperref[LevyLemma]{Levy’s lemma}, we find, for any fixed $U$, $\prob{f_U(\psi)\approx 1}$ is doubly exponentially small in $n$.
Therefore, using the $\epsilon$-net and the union bound argument, we can establish that, with overwhelming probability, no tripartite local unitary maps $|\psi\rangle$ to $|\psi^*\rangle$. 
\end{proof}

\begin{comment}
\begin{proof}[Proof Sketch]
    Suppose $\ket{\psi^*}=U\ket{\psi}$ where $U=U_A\otimes U_B\otimes U_C$, then $\mel{\psi^*}{U}{\psi}=1$.
    Let us consider $f_U(\psi)=\mel{\psi^*}{U}{\psi}$. Taking the Haar average, we have
    \begin{align}
    \EX{ \langle \psi|U | \langle }
    \end{align}
    
    Similar to \cref{eq-fLip}, $f_U(\psi)$ is $O(1)$-Lipshitz.

    Regarding $\mathbb{C}^{d}$ as $\mathbb{R}^{2d}$, we have the Haar average:
    \begin{equation}
        \EX{f_U(\psi)} = \frac{1}{2d}\Tr(K\tilde{U}).
    \end{equation}
    Here $K$ is the complex conjugation; $\tilde{U}$, a $2d\times 2d$ real orthogonal matrix, is the realification of $U$.
    In the standard basis,
    \begin{equation}
        K=\begin{bmatrix}
I & 0 \\
0 & -I
\end{bmatrix},~~
\tilde{U}=\begin{bmatrix}
A & B \\
-B & A
\end{bmatrix},
    \end{equation}
where $U=A+iB$.
Therefore, $\Tr(K\tilde{U})=0$.

Levy's lemma implies $\prob{f_U(\psi)\approx 1}$ is double exponentially small.
Then use $\epsilon$-net and union bound.    
\end{proof}
\end{comment}

\subsection{Absence of global symmetry}

The final application our method concerning the absence of global symmetry. 
Conventionally, a global symmetry is characterized as a tensor-product representation of a group acting on local degrees of freedom. 
A prototypical example is a global $\mathbb{Z}_2$ symmetry on qubits generated by $\otimes_j X_j$, which implements a simultaneous spin flip on every site and corresponds to conservation of total $\mathbb{Z}_2$ charge.

More generally, one may characterize a global symmetry of a pure state as invariance under a tensor-product unitary. Namely, given an $m$-partite state $|\psi\rangle_{R_1 \cdots R_m}$, we say that it possesses a (possibly approximate) global symmetry if there exists a tensor-product unitary
\begin{align}
U = U_{R_1} \otimes \cdots \otimes U_{R_m} \qquad
\text{such that} \qquad
U |\psi\rangle \approx |\psi\rangle,
\end{align}
where at least one $U_{R_{j}}$ is non-trivial (not proportional to the identity) in a sense of \cref{eq:logicalrequire}. 
This is a broad characterization as we do not require the operators to form a finite group, nor do we require exact equality. Nevertheless, this formulation contains the conventional notion of (exact or approximate) global symmetry under finite groups as a special case.

A natural question is whether Haar random states admit any such (approximate) tensor-product symmetry.
The bipartite case again provides useful intuition. A Haar random state $|\psi\rangle_{AB}$ can be brought, by local unitaries, into a tensor product of approximate EPR pairs and a product state. Since each EPR pair is invariant under nontrivial operators such as $X \otimes X$ and $Z \otimes Z$, conjugating these symmetries back by the inverse local unitaries produces nontrivial tensor-product operators $U_A \otimes U_B$ that leave $|\psi\rangle_{AB}$ invariant. Hence bipartite Haar random states generically possess bipartite tensor-product symmetries.

In contrast, this structure disappears in the tripartite setting. 

\begin{theorem}[Absence of global symmetry]\label{thm:global}
Let $|\psi\rangle_{ABC}$ be a Haar random state with $n_A,n_B,n_C < n/2$. Then, with probability $1-o(1)$ in the large-$n$ limit, there does not exist a tensor-product unitary
\begin{align}
U = U_A \otimes U_B \otimes U_C, \qquad \text{such that} \qquad U|\psi\rangle = |\psi\rangle,
\end{align}
with at least one $U_R$ being nontrivial.
\end{theorem}
Here and in \cref{thm:LOGICALgeneral}, nontriviality is defined as in the second inequality of \cref{eq:logicalrequire}.

\begin{proof}[Proof sketch]
Suppose such a tensor-product unitary exists. Then $\langle \psi | U | \psi \rangle = 1$.
Define
\begin{align}
g_U(\psi) := \langle \psi | U | \psi \rangle.
\end{align}
where $g_U$ is $O(1)$-Lipschitz on the unit sphere.
We first compute its Haar average:
\begin{align}
\mathbb{E}_\psi \big[ g_U(\psi) \big] =
\frac{\mathrm{Tr}(U)}{d},
\end{align}
where $d=2^n$. If at least one of $U_A,U_B,U_C$ is nontrivial, then $|\mathrm{Tr}(U)| \leq w d$ and hence
\begin{align}
\mathbb{E}_\psi \bigl[ g_U(\psi) \bigr] = w<1.
\end{align}

Applying \hyperref[LevyLemma]{Levy’s lemma}, we find, for fixed $U$, $\prob{g_U(\psi)\approx 1}$ is doubly exponentially small.
Covering the space of tensor-product unitaries by an $\epsilon$-net and applying the union bound, we can establish that, with overwhelming probability, $|\psi\rangle$ does not have non-trivial tripartite global symmetry.
\end{proof}

An analogous statement holds for Haar random isometries, which extends \cref{thm:LOGICAL}.
\begin{theorem}[Absence of tensor-product logical operator]\label{thm:LOGICALgeneral}
Let $V: C \rightarrow AB$ be a Haar random isometry with $n_C < n_A+n_B$ and $|n_A-n_B|<n_C$. Then, with overwhelming probability, $V$ does not admit a nontrivial logical operator of the form $U_A \otimes U_B$.
\end{theorem}

The proof proceeds in the same manner as \cref{thm:LOGICAL}, by considering the concentration of
\begin{align}
\mathrm{Tr}(V^\dagger (U_A \otimes U_B) V U_C^\dagger)
\end{align}
over the Haar measure and applying an $\epsilon$-net argument.

\section{Holography}

In this section, we illustrate two particular applications of our results in the AdS/CFT correspondence. 
The first question concerns the presence/absence of bipartite entanglement in holographic mixed states.
The second question concerns whether the converse of entanglement wedge reconstruction holds or not. 
Both questions have remained unresolved and led to important conceptual puzzles. 
While we will keep the presentation of this section minimal, we encourage readers to refer to~\cite{Mori:2024gwe} for detailed discussions and background of these problems in the AdS/CFT correspondence. 

\subsection{Entanglement distillation}

In the AdS/CFT correspondence, entanglement entropy $S_{A}$ of a boundary subsystem $A$ is given by the Ryu-Takayanagi (RT) formula 
\begin{align}
S_{A} = \frac{1}{4G_{N}} \min_{\gamma_A} \text{Area}(\gamma_A) + \cdots
\end{align}
at the leading order in $1/G_{N}$ for static geometries where $\gamma_A$ is a bulk surface homologous to $A$ in the asymptotically AdS spacetime, and $G_N$ is the Newton's constant. 
(Recall that $G_N$ is a very small constant.)
This formula predicts that two boundary subsystems $A$ and $B$ can have large mutual information even when they are spatially disconnected on the boundary with a separating subsystem $C$:
\begin{align}
I(A:B) \equiv S_{A} + S_{B}-S_{AB} = O(1/G_{N})
\end{align}
when the minimal surface $\gamma_{AB}$ extends into the bulk and connects $A$ and $B$ with a \emph{connected entanglement wedge}.\footnote{At the leading order in $G_N$, the entanglement wedge of a boundary subregion $R$ is defined as a bulk subregion enclosed by the minimal surface $\gamma_R$ together with the AdS boundary. See Fig.~\ref{fig_reconstruction}(a) for its illustration.} A prototypical example is depicted in Fig.~\ref{fig_connected} for the AdS${_3}$/CFT$_{2}$. 

The entanglement structure in $\rho_{AB}$ remains mysterious. 
For one thing, the mutual information is sensitive to classical correlations such as those in the GHZ state. 
Evidence from quantum gravity thought experiments and toy models~\cite{Nezami:2016zni, Susskind:2014yaa, Dong:2021clv} suggests that correlations in $\rho_{AB}$ in holography are not of classical nature at the leading order in $1/G_{N}$.
Furthermore, a previous work~\cite{Akers:2019gcv} showed that $|\Psi_{ABC}\rangle$ in holography must contain some tripartite entanglement at the leading order in $1/G_{N}$. 

In~\cite{Mori:2024gwe}, we proposed that a holographic mixed state $\rho_{AB}$ does not contain bipartite entanglement when two individual minimal surfaces $\gamma_A, \gamma_B$ do not overlap in the bulk. 
Namely, a connected entanglement wedge does not necessarily imply EPR-like bipartite entanglement. 
The present paper provides supporting evidence for our proposal. 
Recall that a Haar random state serves as a minimal toy model of holography. 
We have $S_{A} \approx \min (n_A, n-n_A )$ at the leading order in $n$, which can be interpreted as the RT-like formula with the area (equals the total number of qubits across the cut) minimization by employing tensor diagrams:
\begin{align}
S_{A} \approx \min \big(\figbox{2.0}{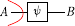} , \figbox{2.0}{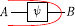} \ \big).  
\end{align}

When $n_A,n_B,n_C < \frac{n}{2}$, we have
\begin{align}
S_{R} \approx  n_R, \qquad (R=A,B,C),
\end{align}
where the minimal surface $\gamma_R$ does not contain the tensor at the center. 
This mimics the situation with a connected entanglement wedge as in Fig.~\ref{fig_connected}. 
Namely, by splitting $C$ into two subsystems, the minimal surface of $AB$ can be schematically depicted as follows
\begin{align}
S_{AB} \approx \ \figbox{2.0}{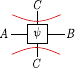} = n_C ,\qquad I(A:B) \approx n_{A} + n_{B} - n_{C} \sim O(n).
\end{align}
Our result shows that no EPR-like entanglement is contained between $A$ and $B$ even when they have a ``connected entanglement wedge''. 

\begin{figure}
\centering
\includegraphics[width=0.25\textwidth]{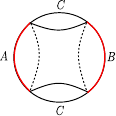}
\caption{Connected entanglement wedge. Here, the minimal area surface of $AB$ is given by geodesics shown in solid lines, leading to a large mutual information $I(A:B)$. Do two subsystems $A$ and $B$ contain EPR-like entanglement?
}
\label{fig_connected}
\end{figure}

While our work provides insights into entanglement properties of holographic mixed states, we hasten to emphasize that whether $\rho_{AB}$ in the real holography possesses EPR-like entanglement or not remains open. 
It is important to recall that Haar random states (and their networks) reproduce entanglement properties of fixed area states, where quantum fluctuations of area operators are strongly suppressed due to the flat spectrum. 
Real holographic states have subleading fluctuations that may potentially lead to bipartite entanglement.
Also, bipartite entanglement in bulk matter fields can contribute to subleading bipartite entanglement in the boundary. 

Our result on tripartite Haar random states is formulated for a finite-dimensional Hilbert space with a tensor product structure $\mathcal H=\mathcal H_A\otimes\mathcal H_B\otimes\mathcal H_C$. 
However, many physically relevant systems, particularly those arising in quantum field theory (QFT), are infinite-dimensional. 
Furthermore, in relativistic QFT, the Hilbert space does not literally factorize into a tensor product for spatially adjacent regions due to the Reeh–Schlieder property and the type-III nature of local von Neumann algebras. 
This naturally leads to the question of how our results should be interpreted in the context of continuum or ``non-qubit'' systems. 

In continuum theories, bipartite entanglement can be unbounded or even divergent. 
For example, the entanglement entropy of a finite spatial region $A$ in relativistic QFT diverges, and two complementary regions $A$ and $B=A^c$ may share an infinite amount of entanglement. 
Such divergences are typically regulated using a UV cutoff (e.g., a lattice spacing or energy cutoff). 
The resulting entanglement is dominated by short-distance modes near the entangling surface, and is accordingly referred to as UV entanglement. 
This contribution indeed contains large amounts of bipartite, EPR-like entanglement between neighboring subregions.

At first sight, this conventional picture of QFT entanglement seems to be at odds with our result, which asserts the absence of EPR-like pairs. 
A natural interpretation is that the Haar random state result concerns the IR, long-distance structure of entanglement, after coarse-graining the UV correlations.
In this sense, the Haar random state may be viewed as a toy model for the entanglement structure that remains after all UV, short-range bipartite correlations have been removed.

This can be seen clearly in the holographic entanglement structure illustrated in Fig.~\ref{fig_connected}. In that geometry, neighboring regions $A$ and $C$, or $B$ and $C$, share large amounts of UV bipartite entanglement. By contrast, spatially separated regions $A$ and $B$ do not exhibit UV-divergent entanglement, although their total entanglement remains parametrically large, of order $1/G_N$. 
The absence of UV contributions is reflected in the fact that the mutual information $I(A:B)$ is independent of the cutoff scale. Equivalently, the minimal surfaces associated with $A$ and $B$ are spatially separated in the bulk, ensuring that short-distance modes near the boundary do not contribute to the $A$-$B$ correlation. 
The resulting entanglement is therefore purely IR in nature.
This IR contribution is encoded geometrically in the central bulk region bounded by the minimal surfaces of $A$, $B$, and $C$.

\subsection{Entanglement wedge reconstruction (and its converse)}

In the AdS/CFT correspondence, the physics of bulk quantum gravity is holographically encoded into boundary quantum systems like a quantum error-correcting code. 
The conceptual pillar behind this interpretation is \emph{entanglement wedge reconstruction}~\cite{Almheiri:2014lwa} which asserts that, \emph{if} a bulk operator $\phi$ lies inside the entanglement wedge $\mathcal{E}_A$ of a boundary subsystem $A$, then it can be expressed as some boundary operator $O_A$ which is supported exclusively on $A$:
\begin{align}\label{eq:ifstatement}
\text{$\phi$ can be reconstructed on $A$} 
\ \Leftarrow \ \text{$\phi$ is inside $\mathcal{E}_A$}. 
\end{align}

While the microscopic mechanisms of the bulk reconstruction still remain somewhat mysterious, random tensor network toy models can provide crucial insights on how the bulk operators may be reconstructed on the boundary subsystems. 
Let us illustrate the idea by using a Haar random state $|\Psi\rangle$ as a minimal toy model. 
Let us first assume that the bulk consists only of one qubit, which is encoded into $n-1$ boundary qubits by viewing an $n$-qubit Haar random state $|\Psi\rangle$ as an encoding isometry $1 \to n-1$ as schematically shown below:
\begin{align}
\figbox{2.0}{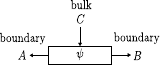}  
\end{align}
where the bulk qubit is denoted by $C$ and the boundary qubits are partitioned into $AB$. 

In this toy model, the question of whether the bulk unitary operator $U_C$ can be reconstructed on a subsystem $A$ is equivalent to whether a logical unitary operator $\overline{U_C}$ can be supported on $A$ in the $C\to AB$ quantum error-correcting code. 
It is well known that $A$ supports a non-trivial logical operator $\overline{U_C}$ when $n_A > \frac{n}{2}$ (and does not support one when $n_A < \frac{n}{2}$). 
\footnote{
This can be done by applying the Petz recovery map~\cite{Barnum:2002bfd}.
}

This standard result on Haar encoding can be understood as entanglement wedge reconstruction. 
Recall that, for static cases in the AdS/CFT correspondence, the entanglement wedge is computed by minimizing the generalized entropy
\begin{align}
S_A = \min_{\gamma_A} \frac{ \text{Area}(\gamma_A) }{4 G_N} + S_{\text{bulk}},
\end{align}
where $S_{\text{bulk}}$ is a bulk entropy on a subregion surrounded by $\gamma_A$. 
When $A$ occupies more than half of the total system, we have
\begin{align}
S_A = 
\figbox{2.0}{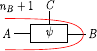}
\end{align}
where the bulk qubit $C$ is inside the $\mathcal{E}_A$, suggesting its recoverability on $A$. 
Here, $+1$ comes from $ S_{\text{bulk}}=S_C = 1$.
On the other hand, when $A$ occupies less than half of the total system, the minimal surface is given by 
\begin{align}
S_A = 
\figbox{2.0}{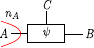}  
\end{align}
where the bulk qubit is outside the $\mathcal{E}_A$. 
Instead, in this case, the bulk qubit is inside $\mathcal{E}_B$, suggesting the recoverability on $B$. 
Hence, the bulk information about $C$ can be recovered from either $A$ or $B$ (unless the sizes of $A$ and $B$ match exactly).

An analogous setup can be considered in the AdS$_{3}$/CFT$_{2}$ correspondence as depicted in Fig.~\ref{fig_reconstruction}(a). 
Here, observe that when the bulk degree of freedom (DOF) $C$ carries a subleading entropies ($O(1)$ or more generally $O(1/G_N^{a})$ with $0<a<1$), the minimal surface $\gamma_A^{EW}$ for defining the entanglement wedge $\mathcal{E}_A$ matches with the minimal area (Ryu-Takayanagi) surface $\gamma_A^{RT}$:
\begin{align}
\gamma_A^{EW} \approx \gamma_A^{RT} \qquad \text{at leading order in $1/G_N$}.
\end{align}
Hence, if the bulk DOF $C$ is enclosed by $\gamma_A^{RT}$, an operator acting on $C$ can be reconstructed on $A$.  
On the other hand, if $C$ lies outside $\gamma_A^{RT}$, it will be enclosed by $\gamma_B^{RT}$ since $\gamma_A^{RT}=\gamma_B^{RT}$ (unless we fine-tune the sizes of $A,B$ so that there are multiple minimal area surfaces).
Thus, an operator acting on $C$ can be reconstructed on $B$.
Recalling the no-cloning theorem, this also implies that an operator on $C$ cannot be reconstructed on $A$ if $C$ lies outside $\gamma_A^{RT}$. 
(Otherwise, quantum information encoded on $C$ could be reconstructed on both $A$ and $B$, creating two copies of $C$.)
Hence, unless the bulk DOF $C$ lies exactly at the minimal surface $\gamma_A^{EW}$, it can always be reconstructed on \emph{one and only one} subsystem, $A$ or $B$. 

\begin{figure}[h]
\centering
a)\raisebox{5.5mm}{\includegraphics[width=0.30\textwidth]{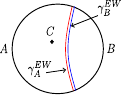}}
b)\includegraphics[width=0.35\textwidth]{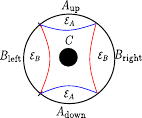}
\caption{
Entanglement wedge reconstruction.
a) When $C$ carries subleading entropy, the minimal surfaces $\gamma_A^{EW}$ and $\gamma_B^{EW}$ coincide. 
Hence, $C$ can be reconstructed on one and only one subsystem $A$ or $B$. 
b) Here, $A=A_{\text{up}} \cup A_{\text{down}}$ and $B=B_{\text{left}}\cup B_{\text{right}}$. 
When $C$ carries leading order entropy, the minimal surfaces $\gamma_A^{EW}$ and $\gamma_B^{EW}$ do not necessarily coincide, and $C$ may not be contained inside $\mathcal{E}_A$ or $\mathcal{E}_B$.  
Can $A$ or $B$ reconstruct $C$? 
}
\label{fig_reconstruction}
\end{figure}

In the above arguments for a Haar random state as well as holography, we have observed that, when the bulk $C$ carries a subleading entropy, entanglement wedge reconstruction is an \emph{if and only if} statement at the leading order in $1/G_N$ (or $n$):
\begin{align}
\text{$\phi$ can be reconstructed on $A$} 
\ \Leftrightarrow \ \text{$\phi$ is inside $\mathcal{E}_A$} \qquad (\text{when $S_{\text{bulk}}$ is subleading}).
\end{align}
A naturally arising question concerns whether this remains the case when $C$ carries a leading order entropy. 
Indeed, it is worth emphasizing that the original entanglement wedge reconstruction \cref{eq:ifstatement} is an \emph{if} statement, implying that a bulk operator can be reconstructed on $A$ \emph{if} it is contained inside $\mathcal{E}_A$. 
Some of previous works shows that the converse statement also holds if the complementary recovery is possible~\cite{Akers:2021fut}.
However, the complementary recovery generally fails when the bulk region $C$ carries leading order entropies.
In such cases, physical or analytical evidence supporting the converse statement has thus far been lacking.
In other words, whether the converse statement holds or not remains unclear:
\begin{align}
\text{$\phi$ can be reconstructed on $A$} 
\ \overset{?}{\Rightarrow} \ \text{$\phi$ is inside $\mathcal{E}_A$} .
\end{align}
When the bulk $C$ is subleading, we were able to promote it to an if and only if statement since $\gamma_A^{EW}=\gamma_B^{EW}$ at the leading order.
However, this may fail when $C$ is not subleading.

One particular holographic setup, highlighting this subtlety, was considered in~\cite{Akers:2020pmf} (Fig.~\ref{fig_reconstruction}(b)). 
Here, the boundary is divided into four segments of roughly equal sizes and organized into $A$ and $B$. 
In the absence of the bulk DOF, the minimal surfaces of $A$ and $B$ would be the same and given by the smaller of the two geodesic lines, colored in red and blue in Fig.~\ref{fig_reconstruction}(b). 
However, when $S_{C} = O(1/G_{N})$, the minimal surface locations may change at the leading order as $S_{\text{bulk}}=S_C$ needs to be included in the evaluation of the generalized entropy.
In particular, we can observe that $\gamma_A^{EW}\not=\gamma_B^{EW}$ at the leading order when the length difference for the red and blue geodesics, multiplied by $1/4G_N$, is smaller than $S_C$ (Fig.~\ref{fig_reconstruction}(b)). 
In this case, the bulk $C$ lies outside $\mathcal{E}_A$ or $\mathcal{E}_B$.~\footnote{ 
The reason why we consider a partition into four segments is that it leaves a sufficiently large bulk region (AdS size) where a bulk DOF with $1/G_N$ entropy can be placed without worrying about backreaction. 
One may replace $C$ with a small (sub-AdS size) black hole or a conical singularity in order to explicitly account for backreaction. We also emphasize that this setup does not require fine-tuning of the sizes of $A$ and $B$ as long as the condition on the two geodesic lengths is met. }

The key question is whether the bulk $C$ is recoverable from $A$ or $B$. 
The converse of entanglement wedge reconstruction would say no since $C$ lies outside $\mathcal{E}_A$ or $\mathcal{E}_B$. 
However, we can observe that $I(C:A), I(C:B) = O(1/G_N)$ in the Choi state, which suggests the presence of leading order correlations between the input $C$ and the output subsystems $A$ and $B$. 
Our result on logical operators supports the converse of entanglement wedge reconstruction as a Haar random encoding with $n_A,n_B,n_C < \frac{n}{2}$ mimics the situation considered in Fig.~\ref{fig_reconstruction}(b); 
\begin{align}
\figbox{2.0}{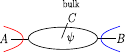}  
\end{align}
where $\mathcal{E}_A$ and $\mathcal{E}_B$ do not contain $C$. 
Theorem~\ref{thm:LOGICAL} from \cref{sec:LOGICAL} then suggests that a non-trivial logical unitary operator $U_C$ cannot be reconstructed on either $A$ or $B$. 

In summary, our result suggests that the converse of entanglement wedge reconstruction is true even when the bulk DOF carries a leading-order entropy. 
One interesting consequence is that there can be a bulk region whose information cannot be reconstructed on either $A$ or $B$. 
In~\cite{Mori:2024gwe}, such a bulk region was referred to as shadow of entanglement wedge. 
This is in strong contrast with random stabilizer states, where the bulk information $C$ can be recovered from either $A$ or $B$ at the leading order. 
Note that it is essential to restrict to unitary logical operators in the above argument, as discussed in section~\ref{sec:LOGICAL}.

\section{Outlook}

In this paper, we showed that a Haar random state does not contain bipartite entanglement under non-trivial tripartition satisfying $n_A, n_B,n_C < \frac{n}{2}$ at the limit of large $n$. 
In the quantum error-correction picture ($C\to AB$), this implies that neither subsystem $A$ nor $B$ can support any logical operator if $n_C < n_A + n_B$ and $|n_A - n_B|< n_C$. 
We also discussed two particular applications of our results in the AdS/CFT correspondence: one about the presence/absence of bipartite entanglement in holographic mixed states and the other about the converse of entanglement wedge reconstruction and the possible extensive bulk region that cannot be reconstructed from bipartite subsystems. 

%It will be interesting to explore the implications of these results in many-body physics and quantum gravity, as well as applications to quantum information processing tasks. 

A natural next step is to consider more structured toy models of holography, such as Haar random tensor networks. 
While Haar random states serve as minimal toy models, tensor networks incorporate geometric locality and are therefore closer in spirit to AdS/CFT.
An important question is whether the absence of bipartite distillable entanglement persists once such geometric structures are introduced.
Technically, this would require understanding whether concentration-of-measure arguments, underlying our results, remain valid in random tensor network ensembles.

The possible breakdown of complementary recovery in holography may also have important implications.
Our result predicts the existence of a semiclassical region that is not contained in the entanglement wedge of any bipartite boundary subsystem.
As an example, in an accompanying work we proposed that the breakdown of complementary recovery for Haar random isometries provides a possible mechanism for the emergence of semiclassical closed baby universes~\cite{Mori:2025jej}.
We expect that this perspective may prove useful in understanding the black hole interior and the role of black hole singularities, where Haar random (or pseudorandom) ideas have already played an important role.

Another important direction is to determine how much randomness is actually required for our conclusions to hold.
Haar randomness is an extreme assumption, and physically relevant states may only approximate it.
It would be interesting to investigate whether similar suppression results hold for states drawn from less random ensembles, such as approximate unitary $k$-designs or $T$-doped Clifford circuits.
Technically, this amounts to asking whether the concentration phenomena underlying our proofs persist under weaker pseudorandomness assumptions, and to identifying the minimal degree of scrambling necessary to eliminate bipartite distillable entanglement.

Furthermore, our analysis places no restriction on the computational complexity of the local operations used in the distillation protocol.
In practice, however, implementable operations are constrained by circuit depth and gate count.
It would therefore be interesting to study entanglement distillation under bounded-complexity restrictions and to understand how computational constraints modify the suppression behavior.
%Clarifying the interplay between randomness, concentration, and circuit complexity may reveal a deeper connection between multipartite entanglement structure and computational scrambling.

\subsection*{Acknowledgment}

We thank the anonymous referees for their stimulating comments.
Research at Perimeter Institute is supported in part by the Government of Canada through the Department of Innovation, Science and Economic Development and by the Province of Ontario through the Ministry of Colleges and Universities. 
This work is supported by JSPS KAKENHI Grant Number 23KJ1154, 24K17047.
This work is supported by the Applied Quantum Computing Challenge Program at the National Research Council of Canada.

\bibliographystyle{JHEP}
\bibliography{ref.bib}

\providecommand{\href}[2]{#2}\begingroup\raggedright\begin{thebibliography}{10}

\bibitem{Srednicki:1994mfb}
M.~Srednicki, \emph{{Chaos and Quantum Thermalization}}, \href{https://doi.org/10.1103/PhysRevE.50.888}{\emph{Phys. Rev. E} {\bfseries 50} (1994) }.

\bibitem{Hayden:2007cs}
P.~Hayden and J.~Preskill, \emph{{Black holes as mirrors: Quantum information in random subsystems}}, \href{https://doi.org/10.1088/1126-6708/2007/09/120}{\emph{JHEP} {\bfseries 09} (2007) 120}.

\bibitem{mehta2004random}
M.L.~Mehta, \emph{Random matrices}, vol.~142, Elsevier (2004).

\bibitem{Page:1993df}
D.N.~Page, \emph{{Average entropy of a subsystem}}, \href{https://doi.org/10.1103/PhysRevLett.71.1291}{\emph{Phys. Rev. Lett.} {\bfseries 71} (1993) 1291}.

\bibitem{Hosur:2015ylk}
P.~Hosur, X.-L.~Qi, D.A.~Roberts and B.~Yoshida, \emph{{Chaos in quantum channels}}, \href{https://doi.org/10.1007/JHEP02(2016)004}{\emph{JHEP} {\bfseries 02} (2016) 004}.

\bibitem{Cotler:2016fpe}
J.S.~Cotler, G.~Gur-Ari, M.~Hanada, J.~Polchinski, P.~Saad, S.H.~Shenker et~al., \emph{{Black Holes and Random Matrices}}, \href{https://doi.org/10.1007/JHEP05(2017)118}{\emph{JHEP} {\bfseries 05} (2017) 118}.

\bibitem{Pastawski:2015qua}
F.~Pastawski, B.~Yoshida, D.~Harlow and J.~Preskill, \emph{{Holographic quantum error-correcting codes: Toy models for the bulk/boundary correspondence}}, \href{https://doi.org/10.1007/JHEP06(2015)149}{\emph{JHEP} {\bfseries 06} (2015) 149}.

\bibitem{Hayden:2016cfa}
P.~Hayden, S.~Nezami, X.-L.~Qi, N.~Thomas, M.~Walter and Z.~Yang, \emph{{Holographic duality from random tensor networks}}, \href{https://doi.org/10.1007/JHEP11(2016)009}{\emph{JHEP} {\bfseries 11} (2016) 009}.

\bibitem{4262758}
A.~Ambainis and J.~Emerson, \emph{Quantum t-designs: t-wise independence in the quantum world},  in \emph{Twenty-Second Annual IEEE Conference on Computational Complexity (CCC'07)}, p.~129, (2007).

\bibitem{Lubkin:1978nch}
E.~Lubkin, \emph{{Entropy of an n-system from its correlation with a k-reservoir}}, \href{https://doi.org/10.1063/1.523763}{\emph{J. Math. Phys.} {\bfseries 19} (1978) 1028}.

\bibitem{Lloyd:1988cn}
S.~Lloyd and H.~Pagels, \emph{Complexity as thermodynamic depth}, \href{https://doi.org/10.1016/0003-4916(88)90094-2}{\emph{Ann. Phys.} {\bfseries 188} (1988) 186}.

\bibitem{Yoshida:2017non}
B.~Yoshida and A.~Kitaev, \emph{{Efficient decoding for the Hayden-Preskill protocol}},  \href{https://arxiv.org/abs/arXiv:1710.03363}{{\ttfamily arXiv:1710.03363}}.

\bibitem{Lu:2020jza}
T.-C.~Lu and T.~Grover, \emph{{Entanglement transitions as a probe of quasiparticles and quantum thermalization}}, \href{https://doi.org/10.1103/PhysRevB.102.235110}{\emph{Phys. Rev. B} {\bfseries 102} (2020) 235110}.

\bibitem{Shapourian:2020mkc}
H.~Shapourian, S.~Liu, J.~Kudler-Flam and A.~Vishwanath, \emph{{Entanglement Negativity Spectrum of Random Mixed States: A Diagrammatic Approach}}, \href{https://doi.org/10.1103/PRXQuantum.2.030347}{\emph{PRX Quantum} {\bfseries 2} (2021) 030347}.

\bibitem{PhysRevA.81.052302}
B.~Yoshida and I.L.~Chuang, \emph{Framework for classifying logical operators in stabilizer codes}, \href{https://doi.org/10.1103/PhysRevA.81.052302}{\emph{Phys. Rev. A} {\bfseries 81} (2010) 052302}.

\bibitem{Nezami:2016zni}
S.~Nezami and M.~Walter, \emph{{Multipartite Entanglement in Stabilizer Tensor Networks}}, \href{https://doi.org/10.1103/PhysRevLett.125.241602}{\emph{Phys. Rev. Lett.} {\bfseries 125} (2020) 241602}.

\bibitem{smith2006typical}
G.~Smith and D.~Leung, \emph{Typical entanglement of stabilizer states}, {\emph{Physical Review A—Atomic, Molecular, and Optical Physics} {\bfseries 74} (2006) 062314}.

\bibitem{Freedman:2016zud}
M.~Freedman and M.~Headrick, \emph{{Bit threads and holographic entanglement}}, \href{https://doi.org/10.1007/s00220-016-2796-3}{\emph{Comm. Math. Phys.} {\bfseries 352} (2017) 407}.

\bibitem{Cui:2018dyq}
S.X.~Cui, P.~Hayden, T.~He, M.~Headrick, B.~Stoica and M.~Walter, \emph{{Bit Threads and Holographic Monogamy}}, \href{https://doi.org/10.1007/s00220-019-03510-8}{\emph{Comm. Math. Phys.} {\bfseries 376} (2019) 609}.

\bibitem{Agon:2018lwq}
C.A.~Ag\'on, J.~De~Boer and J.F.~Pedraza, \emph{{Geometric Aspects of Holographic Bit Threads}}, \href{https://doi.org/10.1007/JHEP05(2019)075}{\emph{JHEP} {\bfseries 05} (2019) 075}.

\bibitem{Harper:2019lff}
J.~Harper and M.~Headrick, \emph{{Bit threads and holographic entanglement of purification}}, \href{https://doi.org/10.1007/JHEP08(2019)101}{\emph{JHEP} {\bfseries 08} (2019) 101}.

\bibitem{Bao:2023til}
N.~Bao and G.~Suer, \emph{{Holographic entanglement distillation from the surface state correspondence}}, \href{https://doi.org/10.1007/JHEP01(2024)091}{\emph{JHEP} {\bfseries 01} (2024) 091}.

\bibitem{Akers:2019gcv}
C.~Akers and P.~Rath, \emph{{Entanglement Wedge Cross Sections Require Tripartite Entanglement}}, \href{https://doi.org/10.1007/JHEP04(2020)208}{\emph{JHEP} {\bfseries 04} (2020) 208}.

\bibitem{Akers:2022max}
C.~Akers, T.~Faulkner, S.~Lin and P.~Rath, \emph{{The Page curve for reflected entropy}}, \href{https://doi.org/10.1007/JHEP06(2022)089}{\emph{JHEP} {\bfseries 06} (2022) 089}.

\bibitem{Akers:2022zxr}
C.~Akers, T.~Faulkner, S.~Lin and P.~Rath, \emph{{Reflected entropy in random tensor networks. Part II. A topological index from canonical purification}}, \href{https://doi.org/10.1007/JHEP01(2023)067}{\emph{JHEP} {\bfseries 01} (2023) 067}.

\bibitem{Akers:2024pgq}
C.~Akers, T.~Faulkner, S.~Lin and P.~Rath, \emph{{Reflected entropy in random tensor networks. Part III. Triway cuts}}, \href{https://doi.org/10.1007/JHEP12(2024)209}{\emph{JHEP} {\bfseries 12} (2024) 209}.

\bibitem{Bravyi:2009zzh}
S.~Bravyi and B.~Terhal, \emph{{A no-go theorem for a two-dimensional self-correcting quantum memory based on stabilizer codes}}, \href{https://doi.org/10.1088/1367-2630/11/4/043029}{\emph{New J. Phys.} {\bfseries 11} (2009) 043029}.

\bibitem{Bennett:1996gf}
C.H.~Bennett, D.P.~DiVincenzo, J.A.~Smolin and W.K.~Wootters, \emph{{Mixed state entanglement and quantum error correction}}, \href{https://doi.org/10.1103/PhysRevA.54.3824}{\emph{Phys. Rev. A} {\bfseries 54} (1996) 3824}.

\bibitem{devetak2005distillation}
I.~Devetak and A.~Winter, \emph{Distillation of secret key and entanglement from quantum states}, {\emph{Proceedings of the Royal Society A: Mathematical, Physical and engineering sciences} {\bfseries 461} (2005) 207}.

\bibitem{Mori:2024gwe}
T.~Mori and B.~Yoshida, \emph{{Does connected wedge imply distillable entanglement?}},  \href{https://arxiv.org/abs/arXiv:2411.03426}{{\ttfamily arXiv:2411.03426}}.

\bibitem{Hayden_2006}
P.~Hayden, D.W.~Leung and A.~Winter, \emph{Aspects of generic entanglement}, \href{https://doi.org/10.1007/s00220-006-1535-6}{\emph{Comm. Math. Phys.} {\bfseries 265} (2006) 95}.

\bibitem{Hayden_2004}
P.~Hayden, D.~Leung, P.W.~Shor and A.~Winter, \emph{Randomizing quantum states: Constructions and applications}, \href{https://doi.org/10.1007/s00220-004-1087-6}{\emph{Comm. Math. Phys.} {\bfseries 250} (2004) 371}.

\bibitem{Böröczky2003}
K.~B{\"o}r{\"o}czky and G.~Wintsche, \emph{Covering the sphere by equal spherical balls},  in \emph{Discrete and Computational Geometry: The Goodman-Pollack Festschrift}, B.~Aronov, S.~Basu, J.~Pach and M.~Sharir, eds., (Berlin, Heidelberg), pp.~235--251, Springer Berlin Heidelberg (2003), \href{https://doi.org/10.1007/978-3-642-55566-4_10}{DOI}.

\bibitem{marchal2017sub}
O.~Marchal and J.~Arbel, \emph{On the sub-gaussianity of the beta and dirichlet distributions}, \href{https://doi.org/10.1214/17-ecp92}{\emph{Electronic Communications in Probability} {\bfseries 22} (2017) }.

\bibitem{nielsen2010quantum}
M.A.~Nielsen and I.L.~Chuang, \emph{Quantum computation and quantum information}, Cambridge university press (2010).

\bibitem{ledoux2001concentration}
M.~Ledoux, \emph{The Concentration of Measure Phenomenon}, no.~89 in Mathematical surveys and monographs, American Mathematical Society (2001).

\bibitem{aubrun2024optimal}
G.~Aubrun, J.~Jenkinson and S.J.~Szarek, \emph{Optimal constants in concentration inequalities on the sphere and in the gauss space}, {\emph{arXiv preprint arXiv:2406.13581} (2024) }.

\bibitem{szarek1997metric}
S.J.~Szarek, \emph{Metric entropy of homogeneous spaces}, {\emph{arXiv preprint math/9701213} (1997) }.

\bibitem{bravyi2007upper}
S.~Bravyi, \emph{Upper bounds on entangling rates of bipartite hamiltonians}, {\emph{Physical Review A—Atomic, Molecular, and Optical Physics} {\bfseries 76} (2007) 052319}.

\bibitem{kim2022chiral}
I.H.~Kim, B.~Shi, K.~Kato and V.V.~Albert, \emph{Chiral central charge from a single bulk wave function}, {\emph{Physical Review Letters} {\bfseries 128} (2022) 176402}.

\bibitem{vardhan2026chirality}
S.~Vardhan, B.~Shi, I.H.~Kim and Y.~Zou, \emph{Chirality, magic, and quantum correlations in multipartite quantum states}, {\emph{SciPost Physics} {\bfseries 20} (2026) 066}.

\bibitem{Susskind:2014yaa}
L.~Susskind, \emph{{ER=EPR, GHZ, and the consistency of quantum measurements}}, \href{https://doi.org/10.1002/prop.201500094}{\emph{Fortsch. Phys.} {\bfseries 64} (2016) 72}.

\bibitem{Dong:2021clv}
X.~Dong, X.-L.~Qi and M.~Walter, \emph{{Holographic entanglement negativity and replica symmetry breaking}}, \href{https://doi.org/10.1007/JHEP06(2021)024}{\emph{JHEP} {\bfseries 06} (2021) 024}.

\bibitem{Almheiri:2014lwa}
A.~Almheiri, X.~Dong and D.~Harlow, \emph{{Bulk Locality and Quantum Error Correction in AdS/CFT}}, \href{https://doi.org/10.1007/JHEP04(2015)163}{\emph{JHEP} {\bfseries 04} (2015) 163}.

\bibitem{Barnum:2002bfd}
H.~Barnum and E.~Knill, \emph{{Reversing quantum dynamics with near-optimal quantum and classical fidelity}}, \href{https://doi.org/10.1063/1.1459754}{\emph{J. Math. Phys.} {\bfseries 43} (2002) 2097}.

\bibitem{Akers:2021fut}
C.~Akers and G.~Penington, \emph{{Quantum minimal surfaces from quantum error correction}}, \href{https://doi.org/10.21468/SciPostPhys.12.5.157}{\emph{SciPost Phys.} {\bfseries 12} (2022) 157} [\href{https://arxiv.org/abs/2109.14618}{{\ttfamily 2109.14618}}].

\bibitem{Akers:2020pmf}
C.~Akers and G.~Penington, \emph{{Leading order corrections to the quantum extremal surface prescription}}, \href{https://doi.org/10.1007/JHEP04(2021)062}{\emph{JHEP} {\bfseries 04} (2021) 062} [\href{https://arxiv.org/abs/2008.03319}{{\ttfamily 2008.03319}}].

\bibitem{Mori:2025jej}
T.~Mori and B.~Yoshida, \emph{{Baby universe as logical qubits: information recovery in random encoding}},  \href{https://arxiv.org/abs/arXiv:2511.20747}{{\ttfamily arXiv:2511.20747}}.

\end{thebibliography}\endgroup
\end{document}